\documentclass{article}
\usepackage[utf8]{inputenc}
\usepackage{amsfonts}
\usepackage{amsthm}
\usepackage{amsmath}
\usepackage{accents}
\usepackage{mathtools}
\usepackage{amssymb}
\usepackage{latexsym}
\usepackage{mathrsfs}
\usepackage{ragged2e}
\usepackage{afterpage}
\usepackage{enumitem}
\usepackage[english]{babel}
\usepackage[a-2b]{pdfx}
\usepackage{bbm}
\usepackage{authblk}
\usepackage{hyperref}
\usepackage[
    type={CC},
    modifier={by},
    version={4.0},
]{doclicense}

\newtheorem{thm}{Theorem}[section]
\numberwithin{equation}{section}
\newtheorem{prop}[thm]{Proposition}
\newtheorem{defn}[thm]{Definition}
\newtheorem{rem}[thm]{Remark}
\newtheorem{lem}[thm]{Lemma}

\title{Ground states for the Hartree energy functional in the critical case}
\author{Tommaso Pistillo}
\affil{Université de Lorraine, CNRS, IECL, F-57000 Metz, France\\Politecnico di Milano, 20133 Milano, Italy}
\date{16 December 2025}

\begin{document}
\maketitle
\textbf{ABSTRACT}: We consider the problem of finding a minimizer $u$ in $ H^1(\mathbb{R}^3)$ for the Hartree energy functional with convolution potential $w$ in $L^\infty(\mathbb{R}^3)+L^{3/2,\infty}(\mathbb{R}^3)$ with $L^\infty$ part vanishing at infinity. This class includes sums of potentials of the kind $-\frac{1}{|x|^\alpha}$, $0<\alpha\le2$, together with the case $w$ in $L^{3/2}(\mathbb{R}^3)$. We prove the existence of such groundstates for a wide range of $L^2$ masses. We also establish basic properties of the groundstates, i.e.~positivity and regularity. Lastly, we exploit the estimates we derived for the stationary problem to prove global well-posedness of the associated evolution problem and orbital stability of the set of ground states.

\section{Introduction}

We consider the Hartree energy functional

\begin{equation}\label{Hartree Energy}
    \mathcal{E}(u)=\frac{1}{2}\int_{\mathbb{R}^3}|\nabla u(x)|^2\,\mathrm{d}x+\frac{1}{4}\iint_{\mathbb{R}^3\times\mathbb{R}^3}|u(x)|^2w(x-y)|u(y)|^2\,\mathrm{d}x\mathrm{d}y,\;u\in H^1(\mathbb{R}^3)
\end{equation}
where $w\not\equiv0$ is a real-valued even function. Minimizers of \eqref{Hartree Energy} are stationary solutions of the time-dependent Hartree equation
\begin{equation}\label{Hartree equation}
    i\partial_t u=-\Delta_x u+\left(w\ast |u|^2\right)u,
\end{equation}
which arises as the mean-field limit for a system of non-relativistic bosons with long-range two-body interaction $w$ which is mostly attractive \cite{Fisica, Lewin Nam, Lieb Thomas}.

Standing wave solutions of \eqref{Hartree equation} also solve
\begin{equation*}
    -\Delta u+\mu u=-\left(w\ast|u|^2\right)u\;\text{ in }\mathbb{R}^3.
\end{equation*}
This generalization of Choquard equation arises from Fröhlich and Pekar's model of the polaron \cite{Frolich 1, Frolich 2, Pekar}, in which electrons and phonons interact in a lattice.
 
The existence of ground states for the Hartree energy has been extensively discussed in the literature. In \cite{Lions} P. L. Lions proved it for the Choquard-Pekar energy functional (i.e.~\eqref{Hartree Energy} with $w(x)=-\frac{1}{|x|}$) for any fixed $L^2$ mass using the concentration-compactness method there developed, instead of the earliest decreasing rearrangement method proposed by Lieb in \cite{Lieb}; more recently, M. Moroz and J. Van Schaftingen \cite{regularity, Review Choquard} extended this result to the optimal choice of parameters $\alpha,p$ for the nonlinear Choquard equation
\begin{equation}\label{nonlinear choq}
    \begin{cases}
        -\Delta u+u=\left(\frac{C_\alpha}{|x|^{d-\alpha}}\ast|u|^p\right)|u|^{p-2}u\;&\text{ in }\mathbb{R}^d\\
        u(x)\rightarrow0\;&\text{ as }|x|\rightarrow\infty,
    \end{cases}
\end{equation}
together with properties of the solution, like smoothness and positivity. Furthermore, N. Ikoma and K. Myśliwy \cite{Ikoma} proved a necessary and sufficient condition on the mass of the ground states in order for them to exist, for a potential $w\in L^{3/2}(\mathbb{R}^3)$. Lastly, one can take the potential $w$ to be nonattractive provided the system is subject to an external potential $V$ which is trapping in some sense (either a local or a global trap) \cite{Fisica, Lions}, or which introduces a kind of spectral gap \cite{Br Fa Pa 1}. In particular, for the Coulomb potential $w(x)=\frac{1}{|x|}$ many results exist on the classes of $V$ which guarantee a ground state \cite{Benguria Brezis, Gustafson Sather, Lebris Lions, Lions 2}.

Although there are plenty of discussions on existence of ground states for Choquard-type equations, there are very few results on uniqueness, especially if no external potential is present; the main results we found of interest were \cite{Lieb}, where Lieb proved uniqueness of the minimizer up to phases and translations for the Coulomb potential $w(x)=-\frac{1}{|x|}$, and \cite{Lenzmann}, where Lenzmann proved uniqueness in $H^{1/2}$ of the ground state to the pseudo-relativistic Hartree equation.

Regarding solutions to the focusing Hartree equation \eqref{Hartree equation}, local existence is well known for the time dependent Choquard equation arising from \eqref{nonlinear choq} (see, for instance, \cite{Ginibre Velo}), while global existence is more delicate and depends on the choice of parameters $\alpha,p$ \cite{Avenia Squassina, Bonanno Avenia, Genev Venkov, Ginibre Velo}. The study of global well-posedness of the Cauchy problem arising from \eqref{Hartree equation}, also comprising the continuous dependence w.r.t.~the initial datum, dates back to \cite{Ginibre Velo 2}.
\subsection{Main results}
In this paper, we work in dimension $3$ for simplicity of exposition but our results can be easily extended to any dimension $d\ge3$.

Our main focus is the study of the existence of minimizers for \eqref{Hartree Energy} with convolution potential $w$ in $L^\infty(\mathbb{R}^3)+L^{3/2,\infty}(\mathbb{R}^3)$ over $\mathcal{S}_\lambda=\{u\in H^1(\mathbb{R}^3)\;:\;\|u\|^2_{L^2}=\lambda\}$;  namely, we are interested in solving
\begin{equation}\label{definition minimizer}
    I(\lambda)=\inf_{u\in\mathcal{S}_\lambda}\mathcal{E}(u).
\end{equation}
Compared to previously cited results, our main contribution consists in considering a large class, probably almost optimal, of sums of potentials in $L^p$ and weak $L^p$ spaces, and a large interval of $\lambda$'s, depending on $w$.

We assume the $L^\infty$ part of $w$ to vanish at infinity; this is a crucial hypothesis, as one can easily prove that \eqref{Hartree Energy} with $w\equiv-1$ has no ground state in $\mathcal{S}_\lambda$ for any $\lambda>0$. Moreover, we assume that the singular part of $w$ is in $L^{3/2,\infty}$, the weak $L^{3/2}$ space endowed with the quasi-norm
\begin{equation*}
    \|f\|_{L^{3/2,\infty}}=\sup_{t>0}\left(t|\{|f|>t\}|^{2/3}\right),
\end{equation*}
where we indicated with $|X|$ the Lebesgue measure of a measurable set $X\subset\mathbb{R}^3$.

We also introduce the following notation regarding Sobolev spaces:
\begin{equation*}
    \begin{split}
        W^{2,r}(\mathbb{R}^3)&=\{u\in L^r(\mathbb{R}^3)\;:\;\Delta u\in L^r(\mathbb{R}^3)\}\\
        \Dot{W}^{m,r}(\mathbb{R}^3)&=\{u\in\mathcal{D}'(\mathbb{R}^3)\;:\;D^mu\in L^r(\mathbb{R}^3)\}\\
    \end{split}
\end{equation*}
where $D u$ is the distributional derivative of $u$.

We prove the following existence result:
\begin{thm}\label{main theorem}
    Let $0\not\equiv w=w_1+w_2\in L^\infty(\mathbb{R}^3)+L^{3/2,\infty}(\mathbb{R}^3)$ be an even function such that there exists $u\in H^1(\mathbb{R}^3)$ for which $\int(w\ast|u|^2)|u|^2<0$ and such that $w_1(x)\xrightarrow{|x|\rightarrow\infty}0$. Define \begin{equation}\label{def c2}
        C_2=\inf\{\|w_2\|_{L^{3/2,\infty}}\,:\,w=w_1+w_2\in L^\infty(\mathbb{R}^3)+L^{3/2,\infty}(\mathbb{R}^3)\}
    \end{equation}
    and
    \begin{equation}\label{def k}
        K=\sup_{\substack{0\ne u\in H^1\\0\ne \tilde{w}\in L^{3/2,\infty}}}\frac{\left|\int(\tilde{w}\ast|u|^2)|u|^2\right|}{\|w\|_{L^{3/2,\infty}}\|u\|^2_{L^2}\|u\|^2_{\Dot{H^1}}}<\infty.
    \end{equation}
    Then, set 
    \begin{equation}
        \lambda_*=\inf\{\lambda>0\,:\,I(\lambda)<0\}
    \end{equation}
    and
    \begin{equation}
        \lambda^*=\frac{1}{C_2K}.
    \end{equation}
    If $\lambda_*<\lambda<\lambda^*$, then problem \eqref{definition minimizer} has a solution $u_*\in\mathcal{S}_\lambda$. Moreover, every minimizer of \eqref{definition minimizer} is positive (up to a constant phase), smooth and in $W^{2,r}(\mathbb{R}^3)$ for every $r\ge2$.
    
    Furthermore, if $w(x)=W(|x|)$ with $W\,:\,(0,\infty)\rightarrow\mathbb{R}$ non-decreasing, then the minimizer can be chosen radial (about some point) and non-increasing.

    Lastly, if $0<\lambda<\lambda_*$, then problem \eqref{definition minimizer} has no solution.
\end{thm}
\begin{rem}\label{first scaling argument}
    It is important to point out that one cannot have a result similar to Theorem \ref{main theorem} for $w\in L^\infty(\mathbb{R}^3)+L^p(\mathbb{R}^3)$ with $p<3/2$, as the resulting functional might not be bounded from below: indeed, letting $w(x)=-\frac{1}{|x|^\alpha}$ with $2<\alpha<3$, we have $w(x)=w\mathbbm{1}_{|x|>R}+w\mathbbm{1}_{|x|\le R}\in L^\infty(\mathbb{R}^3)+L^p(\mathbb{R}^3)$ for some $1\le p<3/2$ with the $L^\infty$ part vanishing at infinity; then, for $u\in H^1(\mathbb{R}^3)$ and $\sigma>0$ let $u_\sigma(x)=\sigma^{-3/2}u(\frac{x}{\sigma})$ and compute
\begin{equation*}
    \mathcal{E}(u_\sigma)=\frac{1}{2\sigma^2}\int_{\mathbb{R}^3}|\nabla u|^2-\frac{1}{4\sigma^\alpha}\iint_{\mathbb{R}^3\times\mathbb{R}^3}|u(x)|^2\frac{1}{|x-y|^\alpha}|u(y)|^2\,\mathrm{d}x\mathrm{d}y\xrightarrow{\sigma\rightarrow0^+}-\infty,
\end{equation*}
so $\mathcal{E}$ is not bounded from below on any $\mathcal{S}_\lambda$.
\end{rem}
\begin{rem}[About $\lambda_*$]\label{rem about lambda_*}
    By an argument similar to the one in Remark \ref{first scaling argument} we can see that
    \begin{equation*}
        \exists\,\lambda>0\;:\; I(\lambda)<0\iff \exists\,u\in H^1\;:\;\int_{\mathbb{R}^3}\left(w\ast|u|^2\right)|u|^2<0,
    \end{equation*}
    hence $\lambda_*<\infty$. Moreover, in the proof of Theorem \ref{main theorem} we also prove that $I(\lambda)<0$ for every $\lambda>\lambda_*$.
    
    It is well known that for some specific short range potentials (e.g.~the Van der Waals-type potentials) we have $\lambda_*>0$; however, it is also known (see, for instance, \cite{regularity, Review Choquard}) that for $w=-\frac{1}{|x|^\alpha}$, $0<\alpha<2$ there exists ground states of any $L^2$ mass. Our framework is compatible with such a result, namely we will show that for such potentials we have. $\lambda_*=0$.
\end{rem}
\begin{rem}[About $\lambda^*$]\label{rem about lambda^*}
While $K$ is a universal constant, $C_2$ depends on $w$ and can vanish; in that case, we set $\lambda^*=\infty$. This is the case, for example, for any potential $w\in L^{3/2}(\mathbb{R}^3)$ and for $w(x)=-\frac{1}{|x|^\alpha}$, $0<\alpha<2$.
\end{rem}
\begin{rem}[About the regularity of the minimizer]\label{rem properties of minimizer}
    If $w$ has no $L^\infty$ part, then we can prove more integrability for the minimizer $u_*$; in Proposition \ref{specific regularity} we prove that if $w\in L^{3/2,\infty}(\mathbb{R}^3)$ then $u_*\in L^1(\mathbb{R}^3)$ and $u_*\in W^{2,r}(\mathbb{R}^3)$ for every $r>1$.
\end{rem}
We also discuss the global well-posedness of the Cauchy problem associated with \eqref{Hartree equation}; in this regard, the main result we prove is

\begin{thm}\label{evo thm}
    Let $0\not\equiv w=w_1+w_2\in L^\infty(\mathbb{R}^3)+L^{3/2,\infty}(\mathbb{R}^3)$ and $u_0\in H^1(\mathbb{R}^3)$ such that 
    \begin{equation}\label{smallness condition}
        K\|u_0\|^2_{L^2}\|w_2\|_{L^{3/2,\infty}}<2,
    \end{equation}
    where $K>0$ is defined in \eqref{def k}. Then there exists a unique $u\in C\left([0,+\infty);H^1\right)\cap C^1\left([0,+\infty);H^{-1}\right)$ solution to \begin{equation}\label{evolution equation}
    \begin{cases}
        i\partial_t u=-\Delta u+(w\ast|u|^2)u\\
        u(0,\cdot)=u_0\in H^1
    \end{cases}
\end{equation}
and the solution depends continuously on the initial datum.

Moreover,
    \begin{itemize}
        \item (Conservation of mass) $\|u(t)\|^2_{L^2}=\|u_0\|^2_{L^2}$ for every $t\ge0$.
        \item (Conservation of energy) $\mathcal{E}(u(t))=\mathcal{E}(u_0)$ for every $t\ge0$.
    \end{itemize}
\end{thm}

\begin{rem}
    If $C_2=\inf\{\|w_2\|_{L^{3/2,\infty}}\;:\;w=w_1+w_2\in L^\infty+L^{3/2,\infty}\}=0$, like in the case $w\in L^{3/2}(\mathbb{R}^3)$, we have global existence for initial data of every mass.
\end{rem}
When we plug the Ansatz $u(t,x)=e^{i\omega t}\psi(x)$, with $\omega\in\mathbb{R}$, into \eqref{evolution equation} we get the eigenvalue problem 
\begin{equation}\label{eigenvalue problem}
    -\Delta\psi-\left(w\ast|\psi|^2\right)\psi=-\omega\psi
\end{equation}
for $\psi\in H^1(\mathbb{R}^3)$. Such solutions, when they exist, are referred to as \textit{Hartree solitons}. We prove that these solitons (whose global existence is guaranteed with $\psi=u_*$, $\|u_*\|^2_{L^2}=\lambda$, $\omega=|I(\lambda)|$, $\lambda_*<\lambda<\lambda^*$ by Theorems \ref{main theorem} and \ref{evo thm}) are also \textit{orbitally stable}, i.e.~if $u_0$ is close to a ground state then the solution $u(t)$ of \eqref{evolution equation} will be close to a ground state for every $t\ge0$, see Theorem \ref{orbital stability thm}.
\subsection{Organization of the paper}
Our discussion is arranged as follows:
\begin{itemize}
    \item In Section 2 we briefly define Lorentz spaces, together with some of their properties; then, we prove the three main inequalities we use throughout this paper, namely \eqref{tech 1}, \eqref{tech 2}, \eqref{tech 3}. We then proceed in describing the variation of the concentration-compactness method we employ for proving the existence of a ground state.
    \item Section 3 is entirely dedicated to the proof of Theorem \ref{main theorem}, first proving some basic properties of the Hartree energy functional \eqref{Hartree Energy} and then applying the aforementioned concentration-compactness method. The last part of the section is devoted to proving positivity and smoothness of the minimizer.
    \item In Section 4 we focus on the dynamical problem \eqref{evolution equation}, first proving global existence of the solution (Theorem \ref{evo thm}) via a classical fixed point argument together with energy estimates, and then proving orbital stability of said solution (Theorem \ref{orbital stability thm}).
\end{itemize}

\section*{Acknowledgements}

This research was funded, in whole or in part, by the Agence Nationale de la Recherche (ANR), project ANR-22-CE92-0013; moreover, we acknowledge the support of the MUR grant "Dipartimento di Eccellenza 2023-2027'' of Dipartimento di Matematica, Politecnico di Milano.

\section{Preliminaries}
In this section, we collect several technical estimates. In the first subsection, we prove some functional estimates in Lorentz spaces that are used in Sections 3 and 4 in a crucial way to control the interaction terms of the Hartree energy functional. In the second subsection, we characterize Lions' concentration-compactness method as done in \cite{Lewin} to better suit with the $H^1$ framework.

\subsection{Functional Inequalities in Lorentz Spaces}\label{2.1}
For $1\le p<\infty$, $1\le q\le\infty$, we define the Lorentz space $L^{p,q}(\mathbb{R}^d)$ as the set of (equivalence classes of) measurable functions $f\,:\,\mathbb{R}^d\rightarrow\mathbb{C}$ such that the following quasi-norm
    \begin{equation*}
        \|f\|_{L^{p,q}}=p^{1/q}\left\|t|\{|f|>t\}|^{1/p}\right\|_{L^q((0,\infty),\mathrm{d}t/t)}
    \end{equation*}
    is finite. We indicated with $|X|$ the Lebesgue measure of a measurable set $X\subset\mathbb{R}^d$.
    
    In particular, for $1\le p<\infty$
    \begin{equation*}
        \|f\|_{L^{p,\infty}}=\sup_{t>0}\left(t|\{|f|>t\}|^{1/p}\right).
    \end{equation*}
Lorentz spaces are a true generalization of the usual Lebesgue spaces: indeed, for every $1<p<\infty$, we can identify $L^{p,p}$ with $L^p$ by the Cavalieri Principle. We also have the following embeddings, reminiscent of the standard $L^p$ ones, see \cite[Proposition 1.4.10]{Grafakos} and \cite[Theorem 7.1]{Peetre}:
\begin{lem}[Inclusion properties]\label{inclusion properties}
    The following inclusions hold:
    \begin{itemize}
        \item $L^{p,q_1}(\mathbb{R}^d)\subset L^{p,q_2}(\mathbb{R}^d)$ for every $1\le p<\infty$, $1\le q_1\le q_2\le\infty$, and the embedding is continuous.
        \item $\Dot{W}^{m,q}(\mathbb{R}^d)\subset L^{p,q}(\mathbb{R}^d)$ with $\frac{1}{p}=\frac{1}{q}-\frac{m}{d}$ for every $1< p<\frac{d}{m}$, and the embedding is continuous.
    \end{itemize}
\end{lem}
Since the Lorentz quasi-norm is invariant under rearrangements of the values of $f$, we can reformulate it as
\begin{equation*}
    \|f\|_{L^{p,q}}=\begin{cases}
        \displaystyle\left(\int_0^\infty(t^{1/p}f^*(t))^q\frac{\mathrm{d}t}{t}\right)^{1/q}\;&\text{if}\;1\le q<\infty\\
        \displaystyle\sup_{t>0} t^{1/p}f^*(t)\;&\text{if}\;q=\infty,
    \end{cases}
\end{equation*}
where $f^*$ is the decreasing rearrangement of $|f|$. Using this reformulation, one can prove that for $1~<~p~<~\infty$, if $f\in L^{p,\infty}$ then for every $\delta>0$ $f\mathbbm{1}_{|f|\ge\delta}\in L^q$ $\forall\;1\le q<p$.

We will use these extensions of the Hölder and Young inequalities to the Lorentz spaces, see \cite{Grafakos, Lemarié, Oneil, Yap}.

\begin{lem}[Hölder Inequality in Lorentz spaces]\label{Hölder Lorentz}
    For $1\le p,p_1,p_2<\infty$, $1\le q,q_1,q_2\le\infty$, there exists a constant $C>0$ such that
    \begin{equation}\label{Hölder}
        \|f_1 f_2\|_{L^{p,q}}\le C\|f_1\|_{L^{p_1,q_1}}\|f_2\|_{L^{p_2,q_2}},\hspace{5mm}\frac{1}{p}=\frac{1}{p_1}+\frac{1}{p_2},\;\frac{1}{q}=\frac{1}{q_1}+\frac{1}{q_2}
    \end{equation}
    whenever the right hand side is finite.
\end{lem}

\begin{lem}[Young Inequality in Lorentz spaces]\label{Young Lorentz}
    For $1< p,p_1,p_2<\infty$, $1\le q,q_1,q_2\le\infty$, there exists a constant $C>0$ such that
    \begin{equation}\label{Young}
        \|f_1\ast f_2\|_{L^{p,q}}\le C\|f_1\|_{L^{p_1,q_1}}\|f_2\|_{L^{p_2,q_2}},\hspace{5mm}1+\frac{1}{p}=\frac{1}{p_1}+\frac{1}{p_2},\;\frac{1}{q}=\frac{1}{q_1}+\frac{1}{q_2}
    \end{equation}
    whenever the right hand site is finite. Moreover, for $1<p<\infty$, $1\le q\le\infty$ there exists $C>0$ such that
    \begin{equation}\label{Young2}
        \|f_1\ast f_2\|_{L^\infty}\le C\|f_1\|_{L^{p,q}}\|f_2\|_{L^{p',q'}},\hspace{5mm}\frac{1}{p}+\frac{1}{p'}=1=\frac{1}{q}+\frac{1}{q'}.
    \end{equation}
\end{lem}

We also have the following estimates, which will be used several times throughout this section. To be more concise, we introduce the following notation: $\|\cdot\|_X\lesssim \|\cdot\|_Y$ iif there exists $C>0$ such that $\|\cdot\|_X\le C\|\cdot\|_Y$. Following the ideas from \cite{Br Fa Pa 1, Br Fa Pa 2}, we prove the following
\begin{lem}[Technical Inequalities]\label{Technical Lemma}
\color{white}.\color{black}
    \begin{enumerate}
        \item Let $u_1$, $u_2\in L^2(\mathbb{R}^3)$ and $w\in L^\infty(\mathbb{R}^3)$. Then
        \begin{equation}\label{tech 1}
            \|w\ast(u_1u_2)\|_{L^\infty}\lesssim\|w\|_{L^\infty}\|u_1\|_{L^2}\|u_2\|_{L^2}.
        \end{equation}
        \item Let $u_1$, $u_2\in\Dot{H^1}(\mathbb{R}^3)$ and $w\in L^{3/2,\infty}(\mathbb{R}^3)$. Then
        \begin{equation}\label{tech 2}
        \|w\ast(u_1u_2)\|_{L^\infty}\lesssim\|w\|_{L^{3/2,\infty}}\|u_1\|_{\Dot{H^1}}\|u_2\|_{\Dot{H^1}}.
        \end{equation}
        \item Let $u_1\in L^2(\mathbb{R}^3)$, $u_2,u_3\in\Dot{H^1}(\mathbb{R}^3)$ and $w\in L^{3/2,\infty}(\mathbb{R}^3)$. Then
        \begin{equation}\label{tech 3}
            \|\left(w\ast(u_1u_2)\right)u_3\|_{L^2}\lesssim\|w\|_{L^{3/2,\infty}}\|u_1\|_{L^2}\|u_2\|_{\Dot{H^1}}\|u_3\|_{\Dot{H^1}}
        \end{equation}
    \end{enumerate}
\end{lem}

\begin{proof}\color{white}.\\
\color{black}
    \eqref{tech 1} follows directly from the classical Young and Hölder inequalities:
\begin{equation*}
\|w\ast(u_1u_2)\|_{L^\infty}\le\|w\|_{L^\infty}\|u_1u_2\|_{L^1}\lesssim\|w\|_{L^\infty}\|u_1\|_{L^2}\|u_2\|_{L^2}.
\end{equation*}
To prove \eqref{tech 2}, we start applying Young and Hölder inequalities \eqref{Young2} and \eqref{Hölder},
\begin{equation*}
    \begin{split}
    \|w\ast(u_1u_2)\|_{L^\infty}&\lesssim\|w\|_{L^{3/2,\infty}}\|u_1u_2\|_{L^{3,1}}\lesssim\|w\|_{L^{3/2,\infty}}\|u_1\|_{L^{6,2}}\|u_2\|_{L^{6,2}}\\
    &\lesssim\|w\|_{L^{3/2,\infty}}\|u_1\|_{\Dot{H^1}}\|u_2\|_{\Dot{H^1}}
    \end{split}
\end{equation*}
as $\Dot{H^1}(\mathbb{R}^3)\subset L^{6,2}(\mathbb{R}^3)$ continuously by Lemma \ref{inclusion properties}.

To prove \eqref{tech 3}, we use twice Hölder inequality \eqref{Hölder} and once Young inequality \eqref{Young2},
\begin{equation*}
    \begin{split}
        \|\left(w\ast(u_1u_2)\right)u_3\|_{L^2}&\lesssim\|w\ast(u_1u_2)\|_{L^{3,\infty}}\|u_3\|_{L^{6,2}}\lesssim\|w\|_{L^{3/2,\infty}}\|u_1u_2\|_{L^{3/2,\infty}}\|u_3\|_{L^{6,2}}\\
        &\lesssim\|w\|_{L^{3/2,\infty}}\|u_1\|_{L^{2,\infty}}\|u_2\|_{L^{6,\infty}}\|u_3\|_{L^{6,2}}\\
        &\lesssim\|w\|_{L^{3/2,\infty}}\|u_1\|_{L^2}\|u_2\|_{L^{6,2}}\|u_3\|_{L^{6,2}}
    \end{split}
\end{equation*}
by Lemma \ref{inclusion properties}. Finally, we can estimate the terms $u_2$ and $u_3$ as we did for the proof of \eqref{tech 2}.
\end{proof}

\subsection{Concentration Compactness Results}

The key result we use for proving the existence of a ground state is Lions' concentration-compactness principle; in this section, we briefly recall the original Concentration-Compactness principle as stated by Lions \cite[Lemma I.1]{Lions} without proving it, and then we adapt it to the $H^1$ framework as in \cite{Lewin}; to do this, we also use the \textit{bubble decomposition} of a sequence, as introduced in \cite{Brezis Coron, Struwe} and later used also in \cite{Gerard}, together with some ideas from \cite{Lieb cc}.
\begin{lem}[Concentration-Compactness Principle]\label{conc comp lemma}
    Let $(\rho_n)_{n\in\mathbb{N}}\subset L^1(\mathbb{R}^d)$ such that $\rho_n\ge0$ and $\|\rho_n\|_{L^1}=\lambda$ where $\lambda>0$ is fixed. Then there exists a subsequence $(\rho_{n_k})_{k\in\mathbb{N}}$ such that one of the following three possibilities occurs:
    \begin{enumerate}
        \item (Compactness) There exists $(y_k)_{k\in\mathbb{N}}\subset\mathbb{R}^d$ such that for every $\varepsilon>0$ there exists $0<R<\infty$ such that
        \begin{equation}\label{ccl comp}
            \int_{B_R(y_k)}\rho_{n_k}\ge\lambda-\varepsilon;
        \end{equation}
        \item (Vanishing) For every $0<R<\infty$ 
        \begin{equation}
\lim_{k\rightarrow\infty}\sup_{y\in\mathbb{R}^d}\int_{B_R(y)}\rho_{n_k}=0;
\end{equation}
        \item (Dichotomy) There exists $0<\alpha<\lambda$ such that for every $\varepsilon>0$ there exists $k_0\in\mathbb{N}$ and non-negative $ \rho_k^{(1)},\rho_k^{(2)}\in L^1(\mathbb{R}^d)$ such that for every $k\ge k_0$
        \begin{equation}\label{ccl dic}
            \begin{cases}
                \left\|\rho_{n_k}-(\rho_k^{(1)}+\rho_k^{(2)})\right\|_{L^1}\le\varepsilon\\
                \left|\alpha-\|\rho_k^{(1)}\|_{L^1}\right|<\varepsilon,\;\left|(\lambda-\alpha)-\|\rho_k^{(2)}\|_{L^1}\right|<\varepsilon\\
                \emph{dist}\left(\emph{supp}(\rho_k^{(1)}),\emph{supp}(\rho_k^{(2)})\right)\rightarrow\infty.
            \end{cases}
        \end{equation}
    \end{enumerate}
\end{lem}

We adapt this to our setting, characterizing the non-compact cases of Lemma \ref{conc comp lemma} using $\rho_n=|u_n|^2$: first, we use the characterization of vanishing sequences bounded in $H^1$ proved in in \cite[Lemma 12]{Lewin}:
    
\begin{lem}[Characterization of vanishing]\label{Vanishing}
    Let $(u_n)_{n\in\mathbb{N}}$ be a bounded sequence in $H^1(\mathbb{R}^3)$. Then $\displaystyle\lim_{n\rightarrow\infty}\sup_{x\in\mathbb{R}^3}\int_{B_R(x)}|u_n|^2=0$ for every $0<R<\infty$ if and only if $u_n\rightarrow0$ strongly in $L^p$ for all $2<p<6$.

\end{lem}

To exploit this, we define the auxiliary functional $\mathcal{E}^\text{van}$ as the original energy functional $\mathcal{E}$ to which we have removed all the terms which go to $0$ as $u_n\rightarrow0$ in $L^p$, $2<p<6$, i.e.
\begin{equation}\label{Definition vanishing energy}
    \mathcal{E}^\text{van}(u)=\|u\|_{\Dot{H}^1}^2.
\end{equation}
Indeed, we have the following

\begin{lem}\label{int vanishes in subcrit Lp}
    Let $w\in L^\infty(\mathbb{R}^3)+L^{3/2,\infty}(\mathbb{R}^3)$ satisfy the hypotheses of Theorem \ref{main theorem} and let $(u_n)_{n\in\mathbb{N}}$ be a bounded sequence in $ H^1(\mathbb{R}^3)$ such that $\|u_n\|^2_{L^2}=\lambda$ for every $n$ and $u_n\rightarrow0$ in $L^p$ for every $p\in(2,6)$. Then
    \begin{equation*}
        \int_{\mathbb{R}^3}\left(w\ast|u_n|^2\right)|u_n|^2\xrightarrow{n\rightarrow\infty}0.
    \end{equation*}
\end{lem}

\begin{proof}
    For $\delta>0$, let $w_{j,\delta}=w\mathbbm{1}_{|w_j|\ge\delta}$ $j=1,2$. Notice that the set $\Omega_\delta=\{x\in\mathbb{R}^3\;:\;|w_1(x)|>\delta\}$ has finite Lebesgue measure $\omega_\delta$ for every $\delta$ since $w_1(x)\rightarrow0$ as $|x|\rightarrow\infty$. Then,
    \begin{equation*}
        \begin{split}
            \left|\iint_{\mathbb{R}^3\times\mathbb{R}^3}|u_n(x)|^2w_1(x-y)|u_n(y)|^2\,\mathrm{d}x\mathrm{d}y \right|&\le \delta\lambda^2+\iint_{|w_1(x-y)|>\delta}|u_n(x)|^2|w_1(x-y)||u_n(y)|^2\,\mathrm{d}x\mathrm{d}y\\
            &=\delta\lambda^2+\int_{\mathbb{R}^3}|u_n(x)|^2\int_{\Omega_\delta}|w_1(z)||u_n(x-z)|^2\,dz\mathrm{d}x\\
            &=\delta\lambda^2+\int_{\mathbb{R}^3}|u_n(x)|^2\|w_1\|_{L^\infty}\|u_n\|^2_{L^2(\Omega_\delta)}\,\mathrm{d}x\\
            &\le\delta\lambda^2+\lambda\|w_1\|_{L^\infty}\|u_n\|^2_{L^2(\Omega_\delta)}\le\delta\lambda^2+\lambda\omega_\delta^{1/6}\|w_1\|_{L^\infty}\|u_n\|^2_{L^3(\Omega_\delta)}\\
            &\le\delta\lambda^2+\lambda\omega_\delta^{1/6}\|w_1\|_{L^\infty}\|u_n\|^2_{L^3(\mathbb{R}^3)}\xrightarrow{n\rightarrow\infty}\delta\lambda^2
        \end{split}
    \end{equation*}
since $L^3(\Omega_\delta)\subset L^2(\Omega_\delta)$ continuously. Similarly, by Hölder inequality, for every $1\le q<3/2$,
\begin{equation*}
    \begin{split}
        \left|\iint_{\mathbb{R}^3\times\mathbb{R}^3}|u_n(x)|^2w_2(x-y)|u_n(y)|^2\,\mathrm{d}x\mathrm{d}y \right|&\le \delta\lambda^2+\iint_{\mathbb{R}^3\times\mathbb{R}^3}|u_n(x)|^2|w_{2,\delta}(x-y)||u_n(y)|^2\mathrm{d}x\mathrm{d}y\\
        &\le\delta\lambda^2+\|u_n\|^4_{L^\frac{4q}{2q-1}}\|w_{2,\delta}\|_{L^q}\xrightarrow{n\rightarrow\infty}\delta\lambda^2,
    \end{split}
\end{equation*}
    where we have used that $w_{2,\delta}\in L^q$ for every $1\le q<3/2$ and that for such $q$ we have $3<\frac{4q}{2q-1}\le 4$, so $u_n\in H^1(\mathbb{R}^3)\subset L^\frac{4q}{2q-1}(\mathbb{R}^3)$. Putting it all together, we get
    \begin{equation*}
        \lim_{n\rightarrow\infty}\int_{\mathbb{R}^3}\left(w\ast|u_n|^2\right)|u_n|^2\le\delta\lambda^2,
    \end{equation*}
    which is enough for us to conclude by arbitrariness of $\delta$.
\end{proof}

We thus define the \textit{minimal vanishing energy}
\begin{equation}
    I^\text{van}(\lambda)=\inf_{u\in\mathcal{S}_\lambda}\mathcal{E}^\text{van}(u)=0,
\end{equation}
so that a minimizing sequence $(u_n)_{n\in\mathbb{N}}\subset \mathcal{S}_\lambda$ can vanish only if $I(\lambda)=I^\text{van}(\lambda)$.

To characterize dichotomy, we exploit \cite[Lemma 6 and Theorem 20]{Lewin} to get the following

\begin{thm}[Characterization of dichotomy]\label{Extracting of a single bubble}
    Let $(u_n)_{n\in\mathbb{N}}$ be a bounded sequence in $H^1(\mathbb{R}^d)$. Then there exists $u^{(1)}\in H^1(\mathbb{R}^d)$ such that for any fixed sequence $0\le R_k\xrightarrow{k\rightarrow\infty}\infty$, there exist a subsequence $(u_{n_k})_{k\in\mathbb{N}}$, sequences of functions $(u_k^{(1)})_{k\in\mathbb{N}}$, $(\psi_k^{(2)})_{k\in\mathbb{N}}$ in $H^1(\mathbb{R}^d)$ and space translations $(x_k^{(1)})_{k\in\mathbb{N}}$ in $\mathbb{R}^d$, such that 
    \begin{equation}\label{convergence in H1}
        \displaystyle\lim_{k\rightarrow\infty}\left\|u_{n_k}- u_k^{(1)}(\cdot-x_k^{(1)})-\psi_k^{(2)}\right\|_{H^1(\mathbb{R}^d)}=0
    \end{equation}
    and such that $u_k^{(1)}$ converges to $u^{(1)}$ weakly in $H^1$ and strongly in $L^p$ for all $2\le p<6$, $\mathrm{supp}(u_k^{(1)})\subset B_{R_k}(0)$ and $\mathrm{supp}(\psi_k^{(2)})\subset\mathbb{R}^d\backslash B_{2R_k}(x_k^{(1)})$ for all $k$.
    
    Moreover,
    \begin{equation}\label{proportionality}
    \begin{split}
        \left\|u_k^{(1)}\right\|_{L^2}\le \|u_{n_k}\|_{L^2}&\text{ and }\;\left\|\psi_k^{(2)}\right\|_{L^2}\le \|u_{n_k}\|_{L^2};\\
        \left\|u_k^{(1)}\right\|_{H^1}\lesssim \|u_{n_k}\|_{H^1}&\text{ and }\;\left\|\psi_k^{(2)}\right\|_{H^1}\lesssim \|u_{n_k}\|_{H^1}.
    \end{split}
\end{equation}
\end{thm}

\begin{rem}
    Theorem \ref{Extracting of a single bubble} gives a a general property of bounded sequences in $H^1$; indeed, it remains true even if dichotomy in the sense of the Concentration-Compactness Lemma does not occur. For a sequence $(u_n)_{n\in\mathbb{N}}$ bounded in $H^1$ with fixed mass $\|u_n\|^2_{L^2}=\lambda$,  dichotomy in the sense of Lemma \ref{conc comp lemma} occurs if and only if  $0<\left\|u^{(1)}\right\|^2_{L^2}<\lambda$.
\end{rem}

\section{Proof of Theorem \ref{main theorem}}

\subsection{Existence of the minimizer}
In this section, we discuss the existence of a minimizer for the Hartree energy functional \eqref{Hartree Energy}, as stated in Theorem \ref{main theorem}. We mainly rely on the concentration-compactness principle \cite{Lions} along with some ideas from \cite{Br Fa Pa 1} and \cite{Lewin}.

We start with a lemma showing the basic properties of the Hartree functional.
\begin{lem}\label{well posedness}
    Let $w\in L^\infty(\mathbb{R}^3)+L^{3/2,\infty}(\mathbb{R}^3)$. Then $\mathcal{E}$ is well defined, translation invariant, continuous on $H^1(\mathbb{R}^3)$ and both bounded from below and coercive on $\mathcal{S}_{\le\lambda}$ for all $0\le\lambda<\lambda^*$, where we defined $\mathcal{S}_{\le\lambda}=\{u\in H^1(\mathbb{R}^3)\,:\,\|u\|^2_{L^2}\le\lambda\}$ and by coercive we mean that there exist $C\in\mathbb{R}$ and $\delta>0$ such that 
    for every $u\in\mathcal{S}_{\le\lambda}$
    \begin{equation*}
        \mathcal{E}(u)\ge C\lambda^2+\delta\|u\|^2_{\Dot{H}^1}.
    \end{equation*}
    
\end{lem}

\begin{proof}
    First of all, $\int_{\mathbb{R}^3}|\nabla u|^2$ is finite for every $u\in H^1$; then, by the technical inequalities \eqref{tech 1} and \eqref{tech 2} we have
\begin{equation}\label{est on convo}
    \left|\int_{\mathbb{R}^3}\left(w\ast|u|^2\right)|u|^2\right|\lesssim\left(\|w_1\|_{L^\infty}\|u\|^2_{L^2}+\|w_2\|_{L^{3/2,\infty}}\|u\|^2_{\Dot{H^1}}\right)\|u\|^2_{L^2},
\end{equation}
which allows us to conclude that $\mathcal{E}$ is well defined on $H^1(\mathbb{R}^3)$.

To prove that $\mathcal{E}$ is continuous from $H^1$ to $\mathbb{R}$, it is sufficient to show that $u\mapsto\int_{\mathbb{R}^3}\left(w\ast |u|^2\right)|u|^2$ is continuous from $H^1$ to $\mathbb{R}$: let $u,\Tilde{u}\in H^1$ then,
\begin{equation*}
    \begin{split}
        \left|\int_{\mathbb{R}^3}\left(w\ast|\Tilde{u}|^2\right)|\Tilde{u}|^2-\int_{\mathbb{R}^3}\left(w\ast|u|^2\right)|u|^2\right|&=\left|\int_{\mathbb{R}^3}\left(w\ast\left(|\Tilde{u}|^2-|u|^2\right)\right)|\Tilde{u}|^2+\int_{\mathbb{R}^3}\left(w\ast|u|^2\right)\left(|\Tilde{u}|^2-|u|^2\right)\right|\\
        &\lesssim\left\|\left(w_1\ast\left(|\Tilde{u}|^2-|u|^2\right)\right)|\Tilde{u}|^2\right\|_{L^1}+\left\|\left(w_2\ast\left(|\Tilde{u}|^2-|u|^2\right)\right)|\Tilde{u}|^2\right\|_{L^1}\\
        &+\left\|\left(w_1\ast|u|^2\right)\left(|\Tilde{u}|^2-|u|^2\right)\right\|_{L^1}+\left\|\left(w_2\ast|u|^2\right)\left(|\Tilde{u}|^2-|u|^2\right)\right\|_{L^1}.
    \end{split}
\end{equation*}
We handle the third and fourth term by applying respectively the technical inequalities \eqref{tech 1} and \eqref{tech 2}:
\begin{equation*}
    \left\|\left(w_1\ast|u|^2\right)\left(|\Tilde{u}|^2-|u|^2\right)\right\|_{L^1}\lesssim\|w_1\|_{L^\infty}\|u\|^2_{L^2}\left\|\left(|\Tilde{u}|^2-|u|^2\right)\right\|_{L^1}
\end{equation*}
and
\begin{equation*}
    \left\|\left(w_2\ast|u|^2\right)\left(|\Tilde{u}|^2-|u|^2\right)\right\|_{L^1}\lesssim\|w_2\|_{L^{3/2,\infty}}\|u\|^2_{\Dot{H}^1}\left\|\left(|\Tilde{u}|^2-|u|^2\right)\right\|_{L^1}.
\end{equation*}
We handle the first two in the same way, first noticing that 
\begin{equation*}
    \int_{\mathbb{R}^3}\left(w\ast\left(|\Tilde{u}|^2-|u|^2\right)\right)|\Tilde{u}|^2=\int_{\mathbb{R}^3}\left(w\ast|\Tilde{u}|^2\right)\left(w\ast\left(|\Tilde{u}|^2-|u|^2\right)\right).
\end{equation*}
Putting all four terms together, we get
\begin{equation*}
    \begin{split}
        \left|\int_{\mathbb{R}^3}(w\ast|\Tilde{u}|^2)|\Tilde{u}|^2-\int_{\mathbb{R}^3}\left(w\ast|u|^2\right)|u|^2\right|&\lesssim(\|w_1\|_{L^\infty}+\|w_2\|_{L^{3/2,\infty}})(\|u\|^2_{H^1}+\|\Tilde{u}\|^2_{H^1})\left\||\Tilde{u}|^2-|u|^2\right\|_{L^1}\\
        &\lesssim(\|w_1\|_{L^\infty}+\|w_2\|_{L^{3/2,\infty}})(\|u\|^2_{H^1}+\|\Tilde{u}\|^2_{H^1})\|\Tilde{u}+u\|_{L^2}\|\Tilde{u}-u\|_{L^2},
    \end{split}
\end{equation*}
which proves continuity.

Lastly, we prove coercivity and boundedness from below on $\mathcal{S}_{\le\lambda}$ for any $0<\lambda<\lambda^*$: we write $\lambda=\lambda^*(1-\delta)$ for some $\delta\in(0,1)$; then, by definition of $C_2$ we can choose a splitting $w=w_1+w_2$ such that $\|w_2\|_{L^{3/2,\infty}}\le C_2\left(1+\frac{\delta}{2-2\delta}\right)=C_2\frac{2-\delta}{2-2\delta}$, so that for any $u\in\mathcal{S}_{\le\lambda}$ we have
\begin{equation*}
    \begin{split}
        \left|\int_{\mathbb{R}^3}\left(w\ast|u|^2\right)|u|^2\right|&\le\lambda^2\|w_1\|_{L^\infty}+K\lambda\|w_2\|_{L^{3/2,\infty}}\|u\|^2_{\Dot{H^1}}\le\lambda^2\|w_1\|_{L^\infty}+KC_2\frac{2-\delta}{2-2\delta}\lambda^*(1-\delta)\|u\|^2_{\Dot{H^1}}\\
        &<\lambda^2\|w_1\|_{L^\infty}+\left(1-\frac{\delta}{2}\right)\|u\|^2_{\Dot{H^1}}
    \end{split}
\end{equation*}
by the technical inequalities \eqref{tech 1}, \eqref{tech 2} and the definition of $K$. This, in turn, implies that

\begin{equation*}
    \begin{split}
        \mathcal{E}(u)&\ge\frac{1}{2}\|u\|^2_{\Dot{H^1}}-\frac{1}{4}\left|\int_{\mathbb{R}^3}\left(w\ast|u|^2\right)|u|^2\right|\ge\frac{1}{2}\|u\|^2_{\Dot{H^1}}-\frac{1}{2}\left|\int_{\mathbb{R}^3}\left(w\ast|u|^2\right)|u|^2\right|\\
        &>\frac{1}{2}\|u\|^2_{\Dot{H^1}}-\frac{1}{2}\left(\|u\|^4_{L^2}\|w_1\|_{L^\infty}+\left(1-\frac{\delta}{2}\right)\|u\|^2_{\Dot{H^1}}\right)=-\frac{\lambda^2}{2}\|w_1\|_{L^\infty}+\frac{\delta}{4}\|u\|^2_{\Dot{H^1}}
    \end{split}
\end{equation*}
so $\mathcal{E}$ is coercive and semibounded from below on $\mathcal{S}_{\le\lambda}$.
\end{proof}

\begin{rem}\label{remark continuity}
    Looking carefully at the proof of continuity of $\mathcal{E}$ w.r.t.~$H^1$ norm, notice that we also proved that $u\mapsto\int_{\mathbb{R}^3}\left(w\ast |u|^2\right)|u|^2$ is uniformly continuous w.r.t.~the $L^2$ norm on all bounded subsets of $H^1$.

    Moreover, $u\mapsto\mathcal{E}(u)$ is uniformly continuous from $H^1$ to $\mathbb{R}$ on bounded subsets of $H^1$.
\end{rem}

From coercivity and local uniform continuity of the energy functional $\mathcal{E}$ follows the lower semicontinuity of its minimal energy $I$; more precisely,

\begin{lem}\label{lsc of gs energy lemma}
    Let $\lambda>0$ and $(\lambda_k)_{k\in\mathbb{N}}$ such that $\lambda_k\xrightarrow{k\rightarrow\infty}\lambda$. Then
    \begin{equation}\label{lsc of gs energy}
        \liminf_{k\rightarrow\infty}I(\lambda_k)\ge I(\lambda).
    \end{equation}
\end{lem}

\begin{proof}
    For $\varepsilon>0$ small, let $u_\varepsilon\in \mathcal{S}_{\lambda-\varepsilon}\cap C^\infty_c$ such that 
    \begin{equation*}
        I(\lambda-\varepsilon)\le\mathcal{E}(u_\varepsilon)\le I(\lambda-\varepsilon)+\varepsilon
    \end{equation*}
    and let $v_\varepsilon\in\mathcal{S}_\varepsilon\cap C^\infty_c$ such that $\|v_\varepsilon\|^2_{\Dot{H}^1}\le\varepsilon$ and $\mathrm{supp}\,u_\varepsilon\,\cap\mathrm{supp}\,v_\varepsilon=\emptyset$, so that $u_\varepsilon+v_\varepsilon\in\mathcal{S}_\lambda$. Since $\mathcal{E}$ is coercive  on $\mathcal{S}_{\le\lambda}$, $(u_\varepsilon)_\varepsilon$ is uniformly bounded in $H^1$, while by its definition so is $(v_\varepsilon)_\varepsilon$. Uniform continuity of $\mathcal{E}$ w.r.t.~the $H^1$ norm on bounded subsets of $H^1$ implies that $\mathcal{E}(u_\varepsilon+v_\varepsilon)=\mathcal{E}(u_\varepsilon)+o_\varepsilon(1)$, so
    \begin{equation*}
        I(\lambda)\le\mathcal{E}(u_\varepsilon+v_\varepsilon)\le\mathcal{E}(u_\varepsilon)+o_\varepsilon(1)\le I(\lambda-\varepsilon)+o_\varepsilon(1)+\varepsilon.
    \end{equation*}
    Passing to the liminf, we obtain $I(\lambda)\le\liminf_{\varepsilon\rightarrow0^+}I(\lambda-\varepsilon)$.
    
    To get the inequality from above, we proceed in a similar way: for $\varepsilon>0$ small, let $u_\varepsilon\in\mathcal{S}_{\lambda+\varepsilon}\cap C^\infty_c$ such that
    \begin{equation*}
        I(\lambda+\varepsilon)\le\mathcal{E}(u_\varepsilon)\le I(\lambda+\varepsilon)+\varepsilon
    \end{equation*}
    and let $\Tilde{u}_\varepsilon=\sqrt{\tfrac{\lambda}{\lambda+\varepsilon}}u_\varepsilon\in\mathcal{S}_\lambda$. Then, letting $v_\varepsilon=u_\varepsilon-\Tilde{u}_\varepsilon$,
    \begin{equation*}
        \|v_\varepsilon\|^2_{L^2}=\left(1-\sqrt{\frac{\lambda}{\lambda+\varepsilon}}\right)\|u_\varepsilon\|^2_{L^2}=\left(\sqrt{\lambda+\varepsilon}-\sqrt{\lambda}\right)^2\le\frac{\varepsilon^2}{4\lambda}
    \end{equation*}
    and
    \begin{equation}
        \|v_\varepsilon\|^2_{\Dot{H}^1}=\left(1-\sqrt{\frac{\lambda}{\lambda+\varepsilon}}\right)\|u_\varepsilon\|^2_{\Dot{H}^1}=\left(\frac{\sqrt{\lambda+\varepsilon}-\sqrt{\lambda}}{\sqrt{\lambda+\varepsilon}}\right)^2\|u_\varepsilon\|^2_{\Dot{H}^1}\le\frac{\varepsilon^2}{4\lambda^2}\|u_\varepsilon\|^2_{\Dot{H}^1}.
    \end{equation}
    
    Once again, since $\mathcal{E}$ is coercive  on $\mathcal{S}_{\le\lambda}$, both $(u_\varepsilon)_\varepsilon$ and $(v_\varepsilon)_\varepsilon$ are uniformly bounded in $H^1$, so the uniform continuity of $\mathcal{E}$ w.r.t.~the $H^1$ norm on bounded subsets of $H^1$ implies that $\mathcal{E}(\Tilde{u}_\varepsilon-v_\varepsilon)=\mathcal{E}(\Tilde{u}_\varepsilon)+o_\varepsilon(1)$, so
    \begin{equation*}
        I(\lambda)\le\mathcal{E}(u_\varepsilon-v_\varepsilon)\le\mathcal{E}(u_\varepsilon)+o_\varepsilon(1)\le I(\lambda+\varepsilon)+o_\varepsilon(1)+\varepsilon.
    \end{equation*}
    Passing to the liminf, we obtain $I(\lambda)\le\liminf_{\varepsilon\rightarrow0^+}I(\lambda+\varepsilon)$.
    
    Putting the two inequalities together yields \eqref{lsc of gs energy}.
\end{proof}

We are now ready to prove Theorem \ref{main theorem}: let $0\not\equiv w\in L^\infty(\mathbb{R}^3)+L^{3/2,\infty}(\mathbb{R}^3)$ such that $w_1(x)\xrightarrow{|x|\rightarrow\infty}0$; since $\mathcal{E}$ is coercive and continuous on $\mathcal{S}_{\le \lambda}$, every minimizing sequence $(u_n)_{n\in\mathbb{N}}\subset \mathcal{S}_\lambda$ of \eqref{definition minimizer} is bounded in $H^1(\mathbb{R}^3)$, so up to a subsequence there exists $u_*\in H^1(\mathbb{R}^3)$ such that $u_n\rightharpoonup u_*$ weakly in $H^1$. We prove that, up to a subsequence, the convergence is also strong in $L^2(\mathbb{R}^3)$.

Applying the Concentration Compactness Principle with our characterization of vanishing, coming from Lemmas \ref{Vanishing} and \ref{int vanishes in subcrit Lp}, and dichotomy, coming from Theorem \ref{Extracting of a single bubble}, to a minimizing sequence $(u_n)_{n\in\mathbb{N}}$ of \eqref{definition minimizer} yields the existence of a subsequence $(u_{n_k})_{k\in\mathbb{N}}$ such that one of the following occurs:
\begin{enumerate}
    \item (Compactness) There exists $(x_k)_{k\in\mathbb{N}}\subset\mathbb{R}^3$ such that $u_{n_k}(\cdot+x_k)\rightarrow u_*$ strongly in $L^2$;
    \item (Vanishing) $I(\lambda)=\displaystyle\lim_{k\rightarrow\infty}\mathcal{E}(u_{n_k})=\lim_{k\rightarrow\infty}\mathcal{E}^{\text{van}}(u_{n_k})=I^\text{van}(\lambda)=0$;
    \item (Dichotomy) There exist $(x_k^{(1)})_{k\in\mathbb{N}}\subset\mathbb{R}^3$, $(u_k^{(1)})_{k\in\mathbb{N}}\subset H^1$, $(\psi_k^{(2)})_{k\in\mathbb{N}}\subset H^1$ and $u^{(1)}\in H^1$ with $0<\|u^{(1)}\|^2_{L^2}<\lambda$ such that 
    \begin{equation*}
        \displaystyle\lim_{k\rightarrow\infty}\left\|u_{n_k}- u_k^{(1)}(\cdot-x_k^{(1)})-\psi_k^{(2)}\right\|_{H^1(\mathbb{R}^d)}=0
    \end{equation*}
    and such that $u_k^{(1)}$ converges to $ u^{(1)}$ weakly in $H^1$ and strongly in $L^p$, $2\le p<6$.
\end{enumerate}
We split the proof of Theorem \ref{main theorem} into four separate claims.

\textbf{CLAIM 1}: Vanishing does not occur for $\lambda>\lambda_*$.
\begin{proof}
In particular, we prove that there exists $\lambda>0$ such that
\begin{equation}\label{no vanishing}
    I(\lambda)<I^\text{van}(\lambda)=0.
\end{equation}
Let $u\in\mathcal{S}_1$ such that $\int_{\mathbb{R}^3}\left(w\ast|u|^2\right)|u|^2<0$; for $\theta>0$ we have that $\theta u\in\mathcal{S}_{\theta^2}$ and
\begin{equation*}
    \mathcal{E}(\theta u)=\frac{\theta^2}{2}\int_{\mathbb{R}^3}|\nabla u|^2+\frac{\theta^4}{4}\int_{\mathbb{R}^3}\left(w\ast|u|^2\right)|u|^2<0\;\text{for}\;\theta\gg1,
\end{equation*}
which shows that there exists $\lambda$ such that \eqref{no vanishing} holds. This proves that vanishing does not occur, and in particular $\lambda_*<\infty$. Notice that this argument we also proves that $I(\lambda)<0$ for every $\lambda>\lambda^*$.
\end{proof}
\begin{rem}
    As anticipated in Remark \ref{rem about lambda_*} if $w(x)=-\frac{1}{|x|^\alpha}$, $0<\alpha<2$, we can prove that $I(\lambda)<0$ for every $\lambda>0$ following the same scaling argument as in Remark \ref{first scaling argument}: for $\sigma>0$, letting $u_\sigma(x)=\sigma^{-3/2}u\left(\frac{x}{\sigma}\right)$, we have
\begin{equation}\label{scaling}
    \mathcal{E}(u_\sigma)=\frac{1}{2\sigma^2}\int_{\mathbb{R}^3}|\nabla u|^2-\frac{1}{4\sigma^\alpha}\int_{\mathbb{R}^3}\left(\frac{1}{|x|^\alpha}\ast|u|^2\right)|u|^2,
\end{equation}
so $\mathcal{E}(u_\sigma)<0$ for $\sigma\gg1$; this proves that for this potential vanishing does not occur for every $\lambda>\lambda_*=0$.
\end{rem}
\textbf{CLAIM 2}: For $\lambda>\lambda_*$, the \textit{binding inequality}
\begin{equation}\label{binding inequality}
    I(\lambda)<I(\alpha)+I(\lambda-\alpha)\;\text{for every}\;0<\alpha<\lambda
\end{equation}
holds.
\begin{proof}
We start by proving that
\begin{equation}\label{scaling for min energy}
    I(\theta\lambda)<\theta I(\lambda)\;\text{for every}\;\lambda>\lambda_*\text{ and  }\theta>1.
\end{equation}
First,
\begin{equation*}
    I(\theta\lambda)=\inf_{u\in\mathcal{S}_\lambda}\left\{\frac{\theta}{2}\|u\|_{\Dot{H}^1}^2+\frac{\theta^2}{4}\int_{\mathbb{R}^3}\left(w\ast|u|^2\right)|u|^2 \right\}=\theta\inf_{u\in\mathcal{S}_\lambda}\left\{\frac{1}{2}\|u\|_{\Dot{H}^1}^2+\frac{\theta}{4}\int_{\mathbb{R}^3}\left(w\ast|u|^2\right)|u|^2 \right\}.
\end{equation*}
Then, notice that when defining problem $\eqref{definition minimizer}$ we can restrict ourselves to taking the inf over the set 
\begin{equation*}
    \mathcal{S}_{\lambda,\beta}=\left\{u\in\mathcal{S}_\lambda\;:\;\int_{\mathbb{R}^3}\left(w\ast|u|^2\right)|u|^2\le-\beta\right\}
\end{equation*}
for some $\beta>0$. Suppose that this is not the case: then, for every minimizing sequence $(v_n)_{n\in\mathbb{N}}\subset\mathcal{S}_\lambda$ we would have $\int_{\mathbb{R}^3}\left(w\ast|v_n|^2\right)|v_n|^2\rightarrow0$. In turn, this would imply that $I(\lambda)=I^\text{van}(\lambda)=0$, which contradicts the assumption $\lambda>\lambda_*$.

To conclude, observe that since $\theta>1$
\begin{equation*}
    I(\theta\lambda)=\frac{\theta}{2}\inf_{u\in\mathcal{S}_{\lambda,\beta}}\left\{\|u\|_{\Dot{H}^1}^2+\frac{\theta}{2}\int_{\mathbb{R}^3}\left(w\ast|u|^2\right)|u|^2 \right\}<\frac{\theta}{2}\inf_{u\in\mathcal{S}_{\lambda,\beta}}\left\{\|u\|_{\Dot{H}^1}^2+\frac{1}{2}\int_{\mathbb{R}^3}\left(w\ast|u|^2\right)|u|^2 \right\}=\theta I(\lambda).
\end{equation*}
We are now ready to prove \eqref{binding inequality}: fix $\lambda>\lambda_*$ and $\alpha\in (0,\lambda)$. Then we must be in one of the following situations (assuming, without loss of generality, that $I(\lambda_*)\le0$ and that $\alpha\ge\lambda-\alpha$):
\begin{enumerate}
    \item $\alpha\in(0,\lambda_*]$ and $\lambda-\alpha\in(0,\lambda_*]$. If this is the case,
    \begin{equation*}
        I(\lambda)<0\le I(\alpha)+I(\lambda-\alpha);
    \end{equation*}
    \item $\alpha\in(\lambda_*,\lambda)$, $\lambda-\alpha\in(0,\lambda_*]$. If this is the case, since $\alpha<\lambda$ and $I(\alpha)<0$, $I(\lambda-\alpha)=0$, by \eqref{scaling for min energy} we have
\begin{equation*}
    I(\lambda)<\frac{\lambda}{\alpha}I(\alpha)<I(\alpha)\le I(\alpha)+I(\lambda-\alpha);
\end{equation*}
    \item $\alpha\in(\lambda_*,\lambda)$ and $\lambda-\alpha\in(\lambda_*,\lambda)$. If this is the case, then by \eqref{scaling for min energy}
    \begin{equation*}
    I(\lambda)<\frac{\lambda}{\alpha}I(\alpha)=I(\alpha)+\frac{\lambda-\alpha}{\alpha}I(\alpha)\le I(\alpha)+I(\lambda-\alpha).
\end{equation*}
\end{enumerate}
\end{proof}
\textbf{CLAIM 3}: Dichotomy does not occur.
\begin{proof}
By Theorem \ref{Extracting of a single bubble}, we know that if we fix a sequence $0\le R_k\xrightarrow{k\rightarrow\infty}\infty$ there exist $u^{(1)}\in H^1(\mathbb{R}^3)$ and
\begin{itemize}
    \item A subsequence $(u_{n_k})_{k\in\mathbb{N}}$,
    \item Sequences of functions $(u^{(1)}_{k})_{k\in\mathbb{N}}$, $(\psi^{(2)}_{k})_{k\in\mathbb{N}}$ in $H^1(\mathbb{R}^3)$,
    \item A sequence of translations $(x_k^{(1)})_{k\in\mathbb{N}}\subset\mathbb{R}^3$
\end{itemize}
such that \eqref{convergence in H1} holds and
\begin{itemize}
    \item $u_k^{(1)}$ converges to $u^{(1)}$ weakly in $H^1(\mathbb{R}^3)$ and strongly in $L^2(\mathbb{R}^3)$,
    \item $\mathrm{supp}(u_k^{(1)})\subset B_{R_k}(0)$ and $\mathrm{supp}(\psi_k^{(2)})\subset\mathbb{R}^3\backslash B_{2R_k}(x_k^{(1)})$.
\end{itemize}
Since our problem is translation invariant, without loss of generality we can choose $x_k^{(1)}\equiv0$.

As $u_k^{(1)}\rightarrow u^{(1)}$ strongly in $L^2$, we have that $\|u_k^{(1)}\|^2_{L^2}\rightarrow\|u^{(1)}\|^2_{L^2}=:\alpha>0$. We remark that $\alpha$ can be assumed non zero because $\alpha=0$ implies vanishing of the minimizing sequence; for more details on why this is true we refer to \cite{Lewin}. This limit, combined with \eqref{convergence in H1} and the reverse triangular inequality, yields $\|\psi_k^{(2)}\|^2_{L^2}\rightarrow\lambda-\alpha$. We also have
\begin{equation}\label{conv energy uk1}
    \liminf_{k\rightarrow\infty}\mathcal{E}(u_k^{(1)})\ge \mathcal{E}(u^{(1)})\ge I(\alpha)
\end{equation}
by Remark \ref{remark continuity} and weak lower semicontinuity of the $L^2$ norm. Similarly, by Lemma \ref{lsc of gs energy lemma} we have
\begin{equation}\label{conv energy psi}
    \liminf_{k\rightarrow\infty}\mathcal{E}(\psi_k^{(2)})\ge \liminf_{k\rightarrow\infty}I(\|\psi_k^{(2)}\|^2_{L^2})\ge I(\lambda-\alpha).
\end{equation}
Now, if we are able to prove that
\begin{equation}\label{split energy}
    \mathcal{E}(u_{n_k})=\mathcal{E}(u_k^{(1)})+\mathcal{E}(\psi_k^{(2)})+o_k(1),
\end{equation}
combining \eqref{conv energy uk1}, \eqref{conv energy psi} and \eqref{split energy} we obtain
\begin{equation*}
    I(\lambda)\ge I(\alpha)+I(\lambda-\alpha),
\end{equation*}
which contradicts the strict energy inequality \eqref{binding inequality} unless $\alpha=\lambda$.

This, together with \eqref{convergence in H1}, implies that $u_{n_k}\rightharpoonup u^{(1)}$ weakly in $L^2$ as $u_k^{(1)}\rightharpoonup u^{(1)}$ weakly in $L^2$ by Theorem \ref{Extracting of a single bubble}. The convergence is also strong in $L^2(\mathbb{R}^3)$ as $\|u_{n_k}\|^2_{L^2}=\lambda=\|u^{(1)}\|^2_{L^2}$, and by uniqueness of the limit $u_*=u^{(1)}$.

To prove \eqref{split energy} we start by noticing that as a consequence of \eqref{convergence in H1} and the continuity of $\mathcal{E}$ we have
\begin{equation*}
    \mathcal{E}(u_{n_k})=\mathcal{E}(u_k^{(1)}+\psi_k^{(2)})+o_k(1).
\end{equation*}
Moreover, since the supports of $u_k^{(1)}$ and $\psi_k^{(2)}$ are disjoint, we have
\begin{equation*}
    \begin{split}
        \mathcal{E}(u_k^{(1)}+\psi_k^{(2)})&=\frac{1}{2}\left\|u_k^{(1)}+\psi_k^{(2)}\right\|^2_{\Dot{H}^1}+\frac{1}{4}\int_{\mathbb{R}^3}(w\ast(|u_k^{(1)}|^2+|\psi_k^{(2)}|^2))(|u_k^{(1)}|^2+|\psi_k^{(2)}|^2)\\
        &=\mathcal{E}(u_k^{(1)})+\mathcal{E}(\psi_k^{(2)})+\frac{1}{4}\int_{\mathbb{R}^3}(w\ast|u_k^{(1)}|^2)|\psi_k^{(2)}|^2+\frac{1}{4}\int_{\mathbb{R}^3}(w\ast|\psi_k^{(2)}|^2)|u_k^{(1)}|^2.
    \end{split}
\end{equation*}
To prove that $\int|u_k^{(1)}(x)|^2w(x-y)|\psi_k^{(2)}(y)|^2\rightarrow0$ as $k\rightarrow\infty$, we proceed in a similar way as in Lemma \ref{int vanishes in subcrit Lp}: defining $w_\delta=w\mathbbm{1}_{|w|\ge\delta}$ for a fixed $\delta>0$, we have
\begin{equation*}
    \begin{split}
        \left|\int_{\mathbb{R}^3}(w\ast|u_k^{(1)}|^2)|\psi_k^{(2)}|^2\right|&\le\delta\lambda^2+\iint_{\mathbb{R}^3\times\mathbb{R}^3}|w_\delta(x-y)||u_k^{(1)}(x)|^2|\psi_k^{(2)}(y)|^2\,\mathrm{d}x\mathrm{d}y\\
    &\le\delta\lambda^2+\iint_{\mathbb{R}^3\times\mathbb{R}^3}|w_\delta(x-y)|\mathbbm{1}_{|x-y|\ge R_k}|u_k^{(1)}(x)|^2|\psi_k^{(2)}(y)|^2\,\mathrm{d}x\mathrm{d}y
    \end{split}
\end{equation*}
since $\|u_k^{(1)}\|^2_{L^2},\|\psi_k^{(2)}\|^2_{L^2}\le\|u_{n_k}\|^2_{L^2}=\lambda$ by \eqref{proportionality} and $\text{dist}(\mathrm{supp}(u_k^{(1)}),\mathrm{supp}(\psi_k^{(2)}))\ge R_k$.

Then, letting $w_{j,\delta}=w_j\mathbbm{1}_{|w_1|\ge\delta}$, $j=1,2$,
\begin{equation*}
    \iint_{\mathbb{R}^3\times\mathbb{R}^3}w_{1,\delta}(x-y)\mathbbm{1}_{|x-y|\ge R_k}|u_k^{(1)}(x)|^2|\psi_k^{(2)}(y)|^2\,\mathrm{d}x\mathrm{d}y\le\|w_{1,\delta}\mathbbm{1}_{|\cdot|\ge R_k}\|_{L^\infty}\lambda^2\rightarrow0
\end{equation*}
as $k\rightarrow\infty$ since $w_1\rightarrow0$ at infinity; Finally, by Hölder inequality we have
\begin{equation*}
\begin{split}
    \iint_{\mathbb{R}^3\times\mathbb{R}^3}w_{2,\delta}(x-y)\mathbbm{1}_{|x-y|\ge R_k}|u_k^{(1)}(x)|^2|\psi_k^{(2)}(y)|^2\,\mathrm{d}x\mathrm{d}y&\le\||u_k^{(1)}|^2\|_{L^\frac{2q}{2q-1}}\||\psi_k^{(2)}|^2\|_{L^\frac{2q}{2q-1}}\|w_{2,\delta}\mathbbm{1}_{|\cdot|\ge R_k}\|_{L^q}\\
    &=\|u_k^{(1)}\|^2_{L^\frac{4q}{2q-1}}\|\psi_k^{(2)}\|^2_{L^\frac{4q}{2q-1}}\|w_{2,\delta}\mathbbm{1}_{|\cdot|\ge R_k}\|_{L^q}\rightarrow0
\end{split}
\end{equation*}
as $k\rightarrow\infty$ for every $1\le q<\frac{3}{2}$. Indeed, $w_{2,\delta}\in L^q$ for $1\le q<\frac{3}{2}$, so $\|w_{2,\delta}\mathbbm{1}_{|\cdot|\ge R_k}\|_{L^q}\rightarrow0$ as $k\rightarrow\infty$. Moreover, for such $q$ we have $3<\frac{4q}{2q-1}\le 4$, so by \eqref{proportionality} and the Sobolev embedding $H^1(\mathbb{R}^3)\subset L^\frac{4q}{2q-1}(\mathbb{R}^3)$
\begin{equation*}
    \|u_k^{(1)}\|^2_{L^\frac{4q}{2q-1}}\|\psi_k^{(2)}\|^2_{L^\frac{4q}{2q-1}}\lesssim\|u_k^{(1)}\|^2_{H^1}\|\psi_k^{(2)}\|^2_{H^1}\lesssim\|u_{n_k}\|_{H^1}^4
\end{equation*}
This gives us \eqref{split energy}, and by the argument mentioned at the beginning of this claim $u_{n_k}\rightarrow u_*$ strongly in $L^2$.
\end{proof}
\textbf{CLAIM 4}: The limit $u_*$ is a minimizer for \eqref{definition minimizer}.
\begin{proof}
Since $u_j\rightarrow u_*$ in $L^2(\mathbb{R}^3)$, $\|u_*\|^2_{L^2}=\lambda$, so $u_*\in\mathcal{S}_\lambda$. Then, by weak lower semicontinuity of the $L^2$ norm we have,
\begin{equation*}
    \|\nabla u_*\|^2_{L^2}\le\liminf_{j\rightarrow\infty}\|\nabla u_j\|^2_{L^2}.
\end{equation*}
Moreover, by Remark \ref{remark continuity}, we have that
\begin{equation*}
    \int_{\mathbb{R}^3}(w\ast|u_*|^2)|u_*|^2=\lim_{j\rightarrow\infty}\int_{\mathbb{R}^3}(w\ast|u_j|^2)|u_j|^2.
\end{equation*}
Combining these, we obtain
\begin{equation*}
    I(\lambda)\le \mathcal{E}(u_*)\le\liminf_{j\rightarrow\infty}\mathcal{E}(u_j)=I(\lambda),
\end{equation*}
so $u_*$ is a minimizer.
\end{proof}

\textbf{CLAIM 5:} For $0<\lambda<\lambda_*$ we have no minimizer for problem \eqref{definition minimizer}.

\begin{proof}
    In the following, we adapt the method proposed in \cite{Ikoma}:

First of all, notice that it is sufficient to prove that if a minimizer $u_\lambda\in\mathcal{S}_\lambda$ for \eqref{definition minimizer} exists, then $I(\lambda')<0$ for every $\lambda'>\lambda$. Then, let $u_\lambda\in\mathcal{S}_\lambda$ such that $I(\lambda)=\mathcal{E}(u_\lambda)$. Since $I(\tilde{\lambda})\le0$ for every $\tilde{\lambda}>0$ and $\|u\|_{\Dot{H}^1}\ge0$ for every $u\in H^1$, we have that $\int(w\ast|u_\lambda|^2)|u_\lambda|^2<0$. Then, writing $V(u)=\frac{1}{4}\int(w\ast |u|^2)|u|^2$, for every $\lambda'>\lambda$ we have
\begin{equation*}
    \begin{split}
        I(\lambda')&\le\frac{1}{2}\left\|\sqrt{\frac{\lambda'}{\lambda}}u_\lambda\right\|^2_{\Dot{H}^1}+V\left(\sqrt{\frac{\lambda'}{\lambda}}u_\lambda\right)=\frac{\lambda'}{2\lambda}\|u\|^2_{\Dot{H}^1}+\frac{\lambda'^2}{\lambda^2}V(u_\lambda)\\
        &=\frac{\lambda'}{\lambda}\left(\frac{1}{2}\|u_\lambda\|^2_{\Dot{H}^1}+V(u_\lambda)+\frac{\lambda'-\lambda}{\lambda}V(u_\lambda)\right)=\frac{\lambda'}{\lambda}\left(I(\lambda)+\frac{\lambda'-\lambda}{\lambda}V(u_\lambda)\right)<0.
    \end{split}
\end{equation*}
thus the existence of a minimizer with $L^2$ mass smaller than $\lambda_*$ would contradict the definition of $\lambda_*$ as the infimum of the $\tilde{\lambda}$ such that $I(\tilde{\lambda})<0$.
\end{proof}

\subsection{Properties of the minimizer}\label{section 2.2}
We now can proceed to prove the properties of the minimizer $u_*$ stated in Remark \ref{rem properties of minimizer}. First of all, since $u_*$ minimizes the energy \eqref{Hartree Energy}, it also solves the eigenvalue equation
\begin{equation}\label{eigen eq}
     -\Delta u+(w\ast |u|^2)u=\omega u,\;u\in H^1(\mathbb{R}^3)
\end{equation}
with $\omega=I(\lambda)<0$. 

We proceed in proving regularity of the minimizer to \eqref{definition minimizer} in the general case $w\in L^\infty+L^{3/2,\infty}$;
\begin{prop}\label{general regularity}
    Let $w\in L^\infty(\mathbb{R}^3)+L^{3/2,\infty}(\mathbb{R}^3)$ satisfy the hypotheses of Theorem \ref{main theorem}, and let $u$ be a solution to \eqref{eigen eq}. Then $u\in C^\infty(\mathbb{R}^3)$ and $u\in W^{2,r}(\mathbb{R}^3)$ for every $2\le r<\infty$.
\end{prop}
\begin{proof}
We prove the result assuming $u\ge0$ a.e. for ease of notation; the extension of the proof to general complex-valued $u$ is straightforward.

As $u\in L^p(\mathbb{R}^3)$, $2\le p\le6$ by Sobolev embedding, we have
\begin{equation*}
    \|(w\ast u^2)u\|_{L^p}\lesssim\|w\ast u^2\|_{L^\infty}\|u\|_{L^p}<\infty
\end{equation*}
by the technical inequalities \eqref{tech 1} and \eqref{tech 2}; then, by the Calderón-Zygmund $L^p$ estimates (see, for example, \cite[Chapter 9]{Regularity}) $u\in W^{2,r}(\mathbb{R}^3)$ for $2\le r\le 6$. Once again, by Sobolev embedding $u\in L^p(\mathbb{R}^3)$ for every $p\ge2$ and so $u\in W^{2,r}(\mathbb{R}^3)$ for every $r\ge2$.

We prove that $u\in C^\infty(\mathbb{R}^3)$ by induction: assuming that $u\in H^{m+1}(\mathbb{R}^3)$ for $m\ge0$, then we compute
\begin{equation*}
    \left\|\partial^m[(w\ast u^2)u]\right\|_{L^2}\le\sum_{k=0}^m\binom{m}{k}\left\|\left(w\ast \partial^k(u^2)\right)\partial^{m-k}u\right\|_{L^2}\le\sum_{k=0}^m\sum_{j=0}^k\binom{m}{k}\binom{k}{j}\left\|\left(w\ast (\partial^ju\,\partial^{k-j}u)\right)\partial^{m-k}u\right\|_{L^2}
\end{equation*}
where we wrote $\partial^k=\partial^k_{x_i}$ for ease of notation. Then, by the technical inequalities \eqref{tech 1} and \eqref{tech 2}
\begin{equation*}
    \begin{split}
        \left\|\left(w\ast (\partial^ju\partial^{k-j}u)\right)\partial^{m-k}u\right\|_{L^2}&\le\left\|w\ast(\partial^ju\partial^{k-j}u)\right\|_{L^\infty}\|\partial^{m-k}u\|_{L^2}\\
        &\lesssim\left(\|w_1\|_{L^\infty}\|\partial^j u\|_{L^2}\|\partial^{k-j} u\|_{L^2}+\|w_2\|_{L^{3/2,\infty}}\|\partial^j u\|_{\Dot{H}^1}\|\partial^{k-j} u\|_{\Dot{H}^1}\right)\|\partial^{m-k} u\|_{L^2}\\
        &=\left(\|w_1\|_{L^\infty}\|u\|_{\Dot{H}^j}\|u\|_{\Dot{H}^{k-j}}+\|w\|_{L^{3/2,\infty}}\|u\|_{\Dot{H}^{j+1}}\|u\|_{\Dot{H}^{k-j+1}}\right)\| u\|_{\Dot{H}^{m-k}}<\infty
    \end{split}
\end{equation*}
so that $\Delta u\in H^m(\mathbb{R}^3)$ and in particular $u\in H^{m+2}(\mathbb{R}^3)$ by standard elliptic regularity, hence by the Sobolev-Morrey embedding $u\in C^\infty(\mathbb{R}^3)$.
\end{proof}
As anticipated in Remark \ref{rem properties of minimizer}, we now prove that we can get more integrability if $w$ has no $L^\infty$ part; to do so, we generalize the method proposed in \cite[Proposition 4.1]{regularity} for $w=\frac{1}{|x|^2}$.

\begin{prop}\label{specific regularity}
    Let $w\in L^{3/2,\infty}(\mathbb{R}^3)$ satisfy the hypotheses of Theorem \ref{main theorem}, and let $u\in H^1(\mathbb{R}^3)$ be a solution of \eqref{eigen eq}. Then $u\in L^1(\mathbb{R}^3)\cap C^\infty(\mathbb{R}^3)$ and $u\in W^{2,r}(\mathbb{R}^3)$ for every $1<r<\infty$.
\end{prop}

\begin{proof}
    Once again, we carry out the proof in the case $u\ge0$ for ease of notation.
    
    $u\in C^\infty(\mathbb{R}^3)$ and $u\in W^{2,r}(\mathbb{R}^3)$ for every $2\le r<\infty$ by Proposition \ref{general regularity}. To prove that $u\in L^1(\mathbb{R}^3)$ and $u\in W^{2,r}(\mathbb{R}^3)$ for every $1<r<2$, we use elliptic bootstrapping:
    
Set $s_0=3$. Then, assume that $u\in L^s$ for every $s\in[s_n,3]$. Then, by the Young inequality \eqref{Young}
\begin{equation*}
    (w\ast u^2)\in L^t\text{ for every }t\text{ such that }\frac{1}{t}=\frac{2}{s}-\frac{1}{3}
\end{equation*}
and by standard Hölder inequality
\begin{equation*}
    (w\ast u^2)u\in L^r\text{ for every }t\text{ such that }\frac{1}{r}=\frac{3}{s}-\frac{1}{3}<1,
\end{equation*}
so $u\in W^{2,r}(\mathbb{R}^3)$ for such $r$ by the Calderón-Zygmund $L^p$ estimates. We have thus proved that if
\begin{equation*}
    \frac{1}{r}=\frac{3}{s}-\frac{1}{3}\text{ and }\frac{1}{s}<\frac{4}{9},
\end{equation*}
then $u\in W^{2,r}(\mathbb{R}^3)$; in other words, if
\begin{equation*}
    \begin{cases}
        \frac{1}{r}<\frac{3}{s_n}-\frac{1}{3}\\
        1<r<3
    \end{cases}
\end{equation*}
then $u\in W^{2,r}(\mathbb{R}^3)$. In turn, this implies that $u\in L^s$ if
\begin{equation*}
    \frac{1}{s}<\frac{3}{s_n}-\frac{1}{3}.
\end{equation*}
Now, since $s_n<3$, we have
\begin{equation*}
    \frac{1}{3}<\frac{1}{s_n}<\frac{3}{s_n}-\frac{1}{3}.
\end{equation*}
Now, if $\frac{3}{s_n}-\frac{1}{3}\ge1$ (so that $\frac{1}{s_n}\ge\frac{4}{9}$) we are done. Otherwise, set $\frac{1}{s_{n+1}}=\frac{3}{s_n}-\frac{1}{3}$ and we are done in a finite number of steps.
\end{proof}

Once we have regularity of the minimizer, we can proceed to prove that all minimizers of \eqref{definition minimizer} are positive up to a phase:

\begin{thm}\label{positivity theorem}
    Let $w\in L^\infty(\mathbb{R}^3)+L^{3/2,\infty}(\mathbb{R}^3)$ satisfy the hypotheses of Theorem \ref{main theorem} and let $u\in\mathcal{S}_\lambda$ be a minimizer for \eqref{definition minimizer}. Then $u$ is of the form
    \begin{equation*}
        u(x)=e^{i\theta}|u(x)|
    \end{equation*}
    for some fixed phase $\theta\in[0,2\pi)$ and $|u(x)|>0$ for every $x\in\mathbb{R}^3$.
\end{thm}

\begin{proof}
    Since $|\nabla|u|(x)|\le|\nabla u(x)|$ a.e., then $\mathcal{E}(|u|)\le\mathcal{E}(u)$ for every $u\in H^1$; thus, if $u$ is a minimizer of \eqref{definition minimizer} so is $|u|\ge0$. Then, by continuity of $|u|$ and the strong maximum principle for second order differential operators we have that $|u|>0$, so all ground states cannot vanish anywhere in $\mathbb{R}^3$. In turn, this implies that $u$ does not vanish and has a constant phase \cite[Lemma 2.10]{Positivity}.
\end{proof}

Finally, we prove radiality of the minimizer, under the additional assumption that $w$ is radial and \textit{non-decreasing} (meaning $w(x)=W(|x|)$ with $W\,:\,(0,\infty)\rightarrow\mathbb{R}$ non-decreasing). First of all, notice that this in particular implies that $w(x)\le0$ a.e. and $w(x)\xrightarrow{|x|\rightarrow\infty}0$. Then, we make use of the \textit{symmetric decreasing rearrangement} of a non-negative measurable function $f$, namely
\begin{equation*}
    f^S(x)=\int_0^\infty\mathbbm{1}_{\{y\,:\,f(y)>t\}^S}(t)\,\mathrm{d}t.
\end{equation*}
where the \textit{symmetric rearrangement} of a measurable set $A\subset\mathbb{R}^d$ is defined as
\begin{equation*}
    A^S=\{x\in\mathbb{R}^d\::\:\omega_d|x|^d\le|A|\};
\end{equation*}
with $\omega_d$ being the volume of the ball of radius 1 in $\mathbb{R}^d$ and $|A|$ is the measure of $A$. We refer to \cite{Lieb Loss} for more details about rearrangements and rearrangement inequalities.

Notice that our assumptions on $w$ imply that $|w|^S=|w|$, as $-w$ is already radial and non-increasing.
\begin{prop}\label{symmetry}
    Let $w\in L^\infty(\mathbb{R}^3)+L^{3/2,\infty}(\mathbb{R}^3)$ be a radial \emph{non-decreasing} function satisfying the hypotheses of Theorem \ref{main theorem}. Then there exist $x_0\in\mathbb{R}^3$ and $v\,:(0,\infty)\rightarrow\mathbb{R}$ non-increasing such that $u(x)=v(|x-x_0|)$ is a minimizer for \eqref{definition minimizer}.
\end{prop}
\begin{proof}
By the Riesz rearrangement inequality (see \cite{Riesz} for the $1$ dimensional case or \cite{Br Li Lo} for the generalization to $\mathbb{R}^d$) we have
\begin{equation*}
    \begin{split}
        \iint_{\mathbb{R}^3\times\mathbb{R}^3}|u(x)|^2|w(x-y)||u(y)|^2\,\mathrm{d}x\mathrm{d}y&\le\iint_{\mathbb{R}^3\times\mathbb{R}^3}(|u|^2)^S(x) |w|^S(x-y)(|u|^2)^S(y)\,\mathrm{d}x\mathrm{d}y\\
        &=\iint_{\mathbb{R}^3\times\mathbb{R}^3}|u^S(x)|^2|w(x-y)||u^S(y)|^2\,\mathrm{d}x\mathrm{d}y,
    \end{split}
\end{equation*}
while by the  Pólya–Szegő inequality \cite[Chapter 7]{Polya Szego} we have
\begin{equation}
    \|\nabla u^S\|_{L^2}\le\|u\|_{\Dot{H}^1}.
\end{equation}
Putting these two together, and recalling that $w\le0$ we get $\mathcal{E}(u^S)\le\mathcal{E}(u)$, which proves that the minimizer can be chosen radial (about some point $x_0$ since $\mathcal{E}$ is translation-invariant) and non-increasing.
\end{proof}

\section{Time dependent Hartee equation}
In this section, we prove all results regarding the Cauchy problem \eqref{evolution equation}. We start with the proof of global existence of a solution as stated in \ref{evo thm}, which relies on a fixed point argument for local existence and a conservation of energy argument for the extension to a global solution; to do so, we use standard techniques (see, for instance, \cite{Evolution} for a general overview or \cite{Evolution pro} for a more complete study). Then, we briefly recall the definition of orbital stability and prove it for the set of ground states of the Hartree equation, similarly to \cite{Fisica}.

\subsection{Proof of Theorem \ref{evo thm}}
\begin{proof}
Let $0\not\equiv w\in L^\infty(\mathbb{R}^3)+L^{3/2,\infty}(\mathbb{R}^3)$ and $u_0\in H^1(\mathbb{R}^3)$ such that \eqref{smallness condition} holds.

We start by proving local existence; first, we know \cite[Lemma 7.1.1]{Evolution} that for $u_0\in H^1$ and $T>0$, $u\in C([0,T];H^1)$ solves \eqref{evolution equation} if and only if it satisfies Duhamel's formula
\begin{equation}\label{integral evolution}
    u(t)=e^{it\Delta}u_0-i\int_0^te^{i(t-s)\Delta}(w\ast|u|^2)(s)u(s)\,\mathrm{d}s.
\end{equation}
We study this as a fixed point equation in $X=C([0,T];H^1)$: we want to apply Banach's fixed point theorem to the function $F\,:\,D\rightarrow X$ defined by
\begin{equation*}
    F(u)(t)=e^{it\Delta}u_0-i\int_0^te^{i(t-s)\Delta}g(u(s))\,\mathrm{d}s,
\end{equation*}
where $g(v)=(w\ast|v|^2)v$, in the closed subset of $X$
\begin{equation*}
    D=\overline{B_1(t\mapsto e^{it\Delta}u_0)}\cap\{u\in X\,:\,u(0)=u_0\}.
\end{equation*}
We start by proving that $F$ well defined: by technical inequalities \eqref{tech 1}, \eqref{tech 2} and \eqref{tech 3}
\begin{equation*}
    \|g(u(s))\|_{L^2}\le\|(w_1\ast|u(s)|^2)u(s)\|_{L^2}+\|(w_2\ast|u(s)|^2)u(s)\|_{L^2}\lesssim\|u(s)\|^3_{H^1}
\end{equation*}
and
\begin{equation*}
    \begin{split}
        \|\nabla(g(u(s)))\|_{L^2}&\lesssim\|w_1\ast \nabla|u(s)|^2\|_{L^\infty}\| u(s)\|_{L^2}+\|(w_2\ast\nabla|u(s)|^2)u(s)\|_{L^2}+\|w\ast |u(s)|^2\|_{L^\infty}\|u(s)\|_{\Dot{H}^1}\\
        &\lesssim\|u(s)\|^3_{H^1}.
    \end{split}
\end{equation*}
We remark that when estimating the second term one has to be careful to have the $L^2$ norm of the gradient when using \eqref{tech 3}, namely
\begin{equation*}
         \|(w_2\ast \nabla|u(s)|^2)u(s)\|_{L^2}=\|(w_2\ast(2u(s)\nabla u(s)))u(s)\|_{L^2}\le2\|w_2\|_{L^{3/2,\infty}}\|u(s)\|_{\Dot{H}^1}^2\|u(s)\|_{\Dot{H}^1}.
\end{equation*}
Then, for $0\le t\le T$,
\begin{equation*}
    \begin{split}
        \|F(u)(t)\|_{H^1}&\le\|e^{it\Delta} u_0\|_{H^1}+\int_0^t\|e^{i(t-s)\Delta}g(u(s))\|_{H^1}\,\mathrm{d}s\\
        &\lesssim \|u_0\|_{H^1}+T\sup_{0\le s\le T}\|u(s)\|^3_{H^1}<\infty
    \end{split}
\end{equation*}
where we have used the unitarity of the semigroup generated by $-\Delta$ and that $u\in X$.

Similarly, for $0\le t_1,t_2\le T$,
\begin{equation*}
    \begin{split}
        \|F(u)(t_2)-F(u)(t_1)\|_{H^1}&\le\|e^{it_2\Delta}u_0-e^{it_1\Delta}u_0\|_{H^1}+\left\|\int_0^{t_2}e^{i(t_2-s)\Delta}g(u(s))\,\mathrm{d}s-\int_0^{t_1}e^{i(t_1-s)\Delta}g(u(s))\,\mathrm{d}s\right\|_{H^1}\\
        &=o_{|t_2-t_1|}(1).
    \end{split}
\end{equation*}
by strong continuity and unitarity of $e^{it\Delta}$, together with technical inequalities \eqref{tech 1}, \eqref{tech 2} and \eqref{tech 3}.
    
Next, we prove that $F(D)\subset D$ for $T$ small enough: clearly $F(u)(0)=u_0$, and
\begin{equation*}
    \|F(u)-(t\mapsto e^{it\Delta}u_0)\|_X=\sup_{0\le t\le T}\|F(u)(t)-e^{it\Delta}u_0\|_{H^1}\le\sup_{0\le t\le T}\int_0^t \|g(u(s))\|_{H^1}\,\mathrm{d}s\lesssim T
\end{equation*}
so for $T$ small enough $F(u)\in B_1(t\mapsto e^{it\Delta}u_0)$.

Finally, we prove that for $T$ small $F$ is a contraction: for $u_1,u_2\in D$,
\begin{equation*}
    \begin{split}
        \|F(u_1)-F(u_2)\|_X&=\sup_{0\le t\le T}\left\|\int_0^t e^{i(t-s)\Delta}(g(u_1(s))-g(u_2(s)))\,\mathrm{d}s \right\|_{H^1}\le\int_0^T\|g(u_1(s))-g(u_2(s))\|_{H^1}\,\mathrm{d}s\\
        &\lesssim \int_0^T \|u_1(s)-u_2(s)\|_{H^1}\,\mathrm{d}s\le T\sup_{0\le s\le T}\|u_1(s)-u_2(s)\|_{H^1}=T\|u_1-u_2\|_X
    \end{split}
\end{equation*}
where we have used the technical inequalities \eqref{tech 1}, \eqref{tech 2}, \eqref{tech 3} as in the proof of well posedness of $F$. Thus, for $T$ small enough $F$ is a contraction and in turn there exists a unique $u\in D$ solution to \eqref{integral evolution}.

The conservation of mass and energy hold for $0\le t\le T$ by simple calculations (see, for instance, \cite[Lemma 7.2.2]{Evolution}. Then, we use the following \cite[Theorem 7.4.1]{Evolution}
\begin{thm}[Maximal time]\label{Maximal Time}
    Let $0\not\equiv w\in L^\infty(\mathbb{R}^3)+L^{3/2,\infty}(\mathbb{R}^3)$. Than there exists a function $T\,:\,H^1(\mathbb{R}^3)\rightarrow (0,+\infty]$ such that for every $u_0\in H^1$ there exists a unique $u\in C\left([0,T(u_0));H^1\right)$ solution to \eqref{evolution equation} for every $0\le T<T(u_0)$. Moreover,
    \begin{itemize}
        \item  If $T(u_0)<+\infty$, then $\displaystyle\lim_{t\rightarrow T(u_0)}\|u(t)\|_{H^1}=+\infty$;
        \item (Conservation of mass) $\|u(t)\|^2_{L^2}=\|u_0\|^2_{L^2}$ for every $0\le t<T(u_0)$;
        \item (Conservation of energy) $\mathcal{E}(u(t))=\mathcal{E}(u_0)$ for every $0\le t<T(u_0)$;
        \item If $u_0\in H^2$, then $u\in C\left([0,T(u_0));H^1\right)\cap C^1\left([0,T(u_0));L^2\right)$.
    \end{itemize}
\end{thm}
This means that in order to prove global existence we just need to prove that $u(t)$ is uniformly bounded in $H^1$. We write 
\begin{equation*}
    \mathcal{E}(u)=\frac{1}{2}\|u\|^2_{\Dot{H^1}}+V(u),\;\;V(u)=\frac{1}{4}\int_{\mathbb{R}^3}\left(w\ast|u|^2\right)|u|^2.
\end{equation*}
Now, fixing $u_0\in H^1$ and letting $\lambda=\|u_0\|^2_{L^2}$,
by conservation of energy we have that 
\begin{equation*}
    \mathcal{E}(u_0)=\frac{1}{2}\|u(t)\|^2_{\Dot{H^1}}+V(u(t))\;\text{for every}\;0\le t<T(u_0),
\end{equation*}
which, together with \eqref{tech 2} and the definition of $K$, yields
\begin{equation*}
        \frac{1}{2}\|u(t)\|_{\Dot{H^1}}^2\le|\mathcal{E}(u_0)|+|V(u(t))|\le |\mathcal{E}(u_0)|+\frac{\lambda^2}{4}\|w_1\|_{L^\infty}+\frac{K\lambda}{4}\|w_2\|_{L^{3/2,\infty}}\|u(t)\|^2_{\Dot{H^1}}.
\end{equation*}
Finally, since $K\lambda\|w_2\|_{L^{3/2,\infty}}<2$ we have
\begin{equation}\label{control of energy}
    \|u(t)\|^2_{\Dot{H^1}}\le\frac{4|\mathcal{E}(u_0)|+\lambda^2\|w_1\|_{L^\infty}}{2-K\lambda\|w_2\|_{L^{3/2,\infty}}},
\end{equation}
which in turn implies global existence of the solution of \eqref{evolution equation}.

The continuity of the solution with respect to the initial datum is ensured by \cite[Proposition 4.3]{Ginibre Velo}.
\end{proof}

\subsection{Orbital stability}

For $\lambda_*<\lambda<\lambda^*$, we define the set of ground states with mass $\lambda$
\begin{equation*}
    M_\lambda=\left\{u\in\mathcal{S}_\lambda\,:\,I(\lambda)=\mathcal{E}(u)\right\},
\end{equation*}
together with the distance function
\begin{equation*}
    d(v)=\inf_{u\in M_\lambda}\|u-v\|_{H^1}.
\end{equation*}
\begin{defn}[Orbital Stability]
    $M_\lambda$ is said \emph{orbitally stable} if for every $\varepsilon>0$ there exists $\delta>0$ such that $d(u_0)<\delta$ implies $d(u(t))<\varepsilon$ for every $t>0$.
\end{defn}

We prove that for every $\lambda$ for which $M_\lambda\ne \emptyset$ and for which there exists a global solution the system is orbitally stable; the original ideas for proving orbital stability for the Hartree equation go back to \cite{Cazenave Lions};
\begin{thm}\label{orbital stability thm}
    Let $w\not \equiv0$ satisfy the hypotheses of Theorem \ref{main theorem}. Then for every $\lambda_*<\lambda<\lambda^*$ such that $K\lambda\|w_2\|_{L^{3/2,\infty}}<2$ orbital stability of $M_\lambda$ holds.
\end{thm}

\begin{proof}
    We start by proving the statement when for initial data $u_0$ such that $\|u_0\|^2_{L^2}=\lambda$. If \eqref{smallness condition} does not hold, there is nothing to prove as we have no global existence. 
    Assuming orbital stability does not hold, then there exists $\varepsilon>0$ and $(u_0^{(n)})_{n\in\mathbb{N}}\subset S_\lambda$ with $d(u_0^{(n)})\xrightarrow{n\rightarrow\infty}0$ and such that, calling $u_n(t)$ the solution to
    \begin{equation*}
    \begin{cases}
        i\partial_t u=-\Delta u+(w\ast|u|^2)u\\
        u(0,\cdot)=u_0^{(n)},
    \end{cases}
\end{equation*}
we have 
\begin{equation}\label{distance bdd}
    d(u_n(t_n))>\varepsilon
\end{equation}
for a suitable sequence of times $(t_n)_{n\in\mathbb{N}}$. Let us denote $v_n=u_n(t_n)$; since both mass and energy are conserved, we have $\|v_n\|^2_{L^2}=\lambda$ and $\mathcal{E}(v_n)=\mathcal{E}(u_n)$, so $(v_n)_{n\in\mathbb{N}}$ is also a minimizing sequence for \eqref{definition minimizer}, hence it converges (in Section 3 we proved that every minimizing sequence converges up to subsequences and translations), contradicting \eqref{distance bdd}.

Then, if the initial datum $u_0'$ is not in $\mathcal{S}_\lambda$, the thesis follows from the continuity w.r.t.~the initial datum of \eqref{evolution equation}: fix $\varepsilon>0$. In the first part, we proved that there exists $\delta_1>0$ such that for every $u_0\in\mathcal{S}_\lambda$ such that  $d(u_0)<\delta_1$ then $d(u(t))<\varepsilon/2$ for every $t>0$, where $u(t)$ is the solution to \eqref{evolution equation} with initial datum $u_0$. Moreover, from the continuity of the solution w.r.t.~the initial datum we get that there exists $\delta_2>0$ such that if $\|u_0-u_0'\|_{H^1}<\delta_2$ then $\|u(t)-u'(t)\|_{H^1}<\varepsilon/2$ for every $t>0$, where $u'(t)$ is the solution to \eqref{evolution equation} with initial datum $u'_0$. Let $\delta=\min\{\delta_1,\delta_2\}$.

Now, for every $u_0'\in H^1\backslash\mathcal{S}_\lambda$ with $d(u_0')<\delta$, there exists $u_0\in\mathcal{S}_\lambda$ such that $d(u_0)<\delta$ and $\|u_0-u_0'\|_{H^1}<\delta$. Finally, we have
\begin{equation*}
    d(u'(t))\le d(u(t))+\|u(t)-u'(t)\|_{H^1}\le \varepsilon/2+\varepsilon/2=\varepsilon
\end{equation*}
for every $t>0$.

Notice that for this second part we didn't comment on the existence of a global solution with initial datum $u'_0\not\in\mathcal{S}_\lambda$; this is because up to taking an even smaller $\delta$, $\|u'_0\|^2_{L^2}\le\lambda+\delta^2<\frac{2}{K\|w_2\|_{L^{3/2,\infty}}}$ so the global existence condition \eqref{smallness condition} still holds.
\end{proof}

\small DIPARTIMENTO DI MATEMATICA, POLITECNICO DI MILANO, 20133 MILANO\\
\textit{e-mail address:} \href{mailto:tommaso.pistillo@polimi.it}{tommaso.pistillo@polimi.it}

UNIVERSITÉ DE LORRAINE, CNRS, IECL, F-57000 METZ, FRANCE\\
\textit{e-mail address:} \href{mailto:tommaso.pistillo@univ-lorraine.fr}{tommaso.pistillo@univ-lorraine.fr}

\doclicenseThis

\begin{thebibliography}{}

\bibitem{Avenia Squassina}  D’Avenia P., Squassina, M., \textit{Soliton dynamics for the Schrödinger–Newton system}, Math. Models Methods Appl. Sci. 24 (2014), no. 3, pp. 553–572.

\bibitem{Benguria Brezis} Benguria, R., Brezis, H., Lieb, E.H., \textit{The Thomas-Fermi-von Weizsäcker theory of atoms and molecules}, Comm. Math. Phys. 79 (1981), no. 2, 167–180.

\bibitem{Bonanno Avenia} Bonanno, C., d’Avenia, P., Ghimenti, M., Squassina, M., \textit{Soliton dynamics for the generalized Choquard equation}, J. Math. Anal. Appl. 417 (2014), no. 1, 180–199.

\bibitem{Br Li Lo}  Brascamp, H.J.; Lieb, Elliott H.; Luttinger, J.M., \textit{A general rearrangement inequality for multiple integrals} Journal of Functional Analysis. (1974) 17: 227–237. 

\bibitem{Br Fa Pa 1} Breteaux, S., Faupin, J., Payet, J., \textit{Quasi-classical ground states. I. Linearly coupled Pauli–Fierz Hamiltonians}  (2023). Documenta Mathematica. 28. 1191-1233. 10.4171/dm/929. 

\bibitem{Br Fa Pa 2} Breteaux, S., Faupin, J., Payet, J., \textit{Quasi-Classical Ground States. II. Standard Model of Non-Relativistic QED}, Annales de l'Institut Fourier, Volume 75 (2025) no. 3, pp. 1177-1220.

\bibitem{Brezis Coron} Brezis, H., Coron, J.-M., \textit{Convergence of solutions of H-systems or how to blow bubbles}, Arch. Rational Mech. Anal., 89 (1985), pp. 21-56.

\bibitem{Evolution pro} Cazenave, T., \textit{Semilinear Schrödinger equations}, Courant Lecture Notes in Mathematics, 10. New York University, Courant Institute of Mathematical Sciences, AMS.

\bibitem{Evolution} Cazenave, T., Haraux, A., \textit{An introduction to Semilinear Evolution Equations}, volume 13 of Oxford Lecture Series in Mathematics and its Applications, Clarendon Press, 2nd edition (2006).

\bibitem{Cazenave Lions} Cazenave, T., Lions, P.-L., \textit{Orbital stability of standing waves for some nonlinear Schrödinger equations}, Comm. Math. Phys. 85 (1982), pp.549–561.

\bibitem{Positivity} Cingolani, S., Secchi, S., Squassina, M., \textit{Semiclassical limit for Schrödinger equations with magnetic field and Hartree-type nonlinearities}, Proceedings of the Royal Society of Edinburgh: Section A Mathematics. 2010;140(5):973-1009.

\bibitem{Frolich 1} Fröhlich, H., \textit{Theory of electrical breakdown in ionic crystal}, Proc. Roy. Soc. Ser. A 160 (1937), no. 901, pp. 230–241.

\bibitem{Frolich 2} Fröhlich, H., \textit{Electrons in lattice fields}, Adv. in Phys. 3 (1954), no. 11.

\bibitem{Fisica} Fröhlich, J., Lenzmann. E., \textit{Mean-Field Limit of Quantum Bose Gases and Nonlinear Hartree Equation}, Séminaire Équations aux dérivées partielles (Polytechnique) dit aussi "Séminaire Goulaouic-Schwartz" (2003-2004), Talk no. 18, 26.

\bibitem{Genev Venkov} Genev, H., Venkov, G.,\textit{ Soliton and blow-up solutions to the time-dependent Schrödinger–Hartree equation}, Discrete Contin. Dyn. Syst. Ser. S 5 (2012), no. 5, 903–923.

\bibitem{Gerard} Gérard, P., \textit{Description du défaut de compacité de l'injection de Sobolev}, ESAIM Control Optim. Calc. Var., 3 (1998), pp. 213-233 (electronic).

\bibitem{Regularity} Gilbarg, D., Trudinger, N.S. \textit{Elliptic Partial Differential Equations of Second Order}, Second Edition, Grundlehren der Mathematischen Wissenschaften, vol. 224, Springer, Berlin, (1983).

\bibitem{Ginibre Velo} Ginibre, J., Velo, G., \textit{On a class of non linear Schrödinger equations with non local interaction}, Math Z 170 (1980), pp. 109–136.

\bibitem{Ginibre Velo 2} Ginibre, J., Velo, G., \textit{On a class of nonlinear Schrödinger equations. I. The Cauchy problem, general case}, Journal of Functional Analysis, Volume 32, Issue 1, (1979), pp. 1-32.

\bibitem{Grafakos} Grafakos, L., \textit{Classic Fourier Analysis}, volume 249 of \textit{Graduate Texts in Mathematics}. New York, NY: Springer, 3rd edition (2014).

\bibitem{Gustafson Sather} K. Gustafson and D. Sather, \textit{A branching analysis of the Hartree equation}, Rend. Mat. (6) 4 (1971), pp. 723–734.

\bibitem{Hunt} Hunt, R. A., \textit{On L(p, q) spaces}, Enseignement Math. (2) 12 (1966), pp. 249-276.

\bibitem{Ikoma} Ikoma, N., Myśliwy, K., \textit{Existence and order of the self–binding transition in non–local non–linear Schrödinger equations} arXiv preprint arXiv:2504.06988 (2025).

\bibitem{Lebris Lions} Le Bris, C., Lions, P.-L., \textit{From atoms to crystals: a mathematical journey}, Bull. Amer. Math. Soc. (N.S.) 42 (2005), no. 3, pp. 291–363.

\bibitem{Lemarié} Lemarié-Rieusset, P. G., \textit{Recent developments in the Navier-Stokes problem}, volume 431 of Chapman Hall/CRC Res. Notes Math. Boca Raton, FL: Chapman \& Hall/CRC, 2002.

\bibitem{Lenzmann} Lenzmann, E., \textit{Uniqueness of ground states for pseudorelativistic Hartree equations}, Analysis \& amp; PDE, vol. 2, no. 1, (2009) pp. 1–27.

\bibitem{Lewin} Lewin, M., \textit{Describing lack of compactness in Sobolev spaces}. Master. Variational Methods in Quantum Mechanics, France. (2010). ff hal-02450559v3f .

\bibitem{Lewin Nam} Lewin, M., Nam, P. T., Rougerie, N., \textit{Derivation of Hartree’s theory for generic mean-field Bose systems}, Adv. Math. 254 (2014), pp. 570–621.

\bibitem{Lieb} Lieb, E. H., \textit{Existence and uniqueness of the minimizing solution of Choquard’s nonlinear equation}. Stud. Appl. Math. 57 (1977), pp. 93–105.

\bibitem{Lieb cc} Lieb, E. H., \textit{On the lowest eigenvalue of the Laplacian for the intersection of two domains}, Invent. Math., 74 (1983), pp. 441-448.

\bibitem{Lieb Loss} Lieb, E., Loss, M., \textit{Analysis}. Graduate Studies in Mathematics. Vol. 14 (2nd ed.). (2001) American Mathematical Society.

\bibitem{Lieb Thomas} Lieb, E. H., Thomas, L. \textit{Exact Ground State Energy of the Strong-Coupling Polaron}, Comm Math Phys 183 (1997), pp. 511–519.

\bibitem{Lions} Lions, P.-L., \textit{The concentration-compactness principle in the calculus of variations. The locally compact case, part 1}. Annales de l'I.H.P. Analyse non linéaire, Volume 1 no. 2, (1984) pp. 109-145.

\bibitem{Lions 2} Lions, P.-L., \textit{Solutions of Hartree–Fock equations for Coulomb systems}, Comm. Math. Phys. 109, no. 1, (1987) pp. 33–97.

\bibitem{regularity} Moroz, M., Van Schaftingen, J., \textit{Ground states of nonlinear choquard equations:existence, qualitative properties and decay asymptotics} J. Funct. Anal. 265 (2013), pp. 153-184.

\bibitem{Review Choquard} Moroz, M., Van Schaftingen, J., \textit{A guide to the Choquard equation}, J. Fixed Point Theory Appl. 19, (2017) pp. 773–813.

\bibitem{Oneil} O’Neil, R., \textit{Convolution operators and L(p,q) spaces}, Duke Math. J., 30 (1963), pp. 129–142.

\bibitem{Peetre}
Peetre, J., \textit{Espaces d'interpolation et théorème de Soboleff}. Annales de l'Institut Fourier, Volume 16 (1966) no. 1, pp. 279-317.

\bibitem{Pekar} Pekar, S., \textit{Untersuchung über die Elektronentheorie der Kristalle}, Akademie Verlag, Berlin (1954).

\bibitem{Polya Szego} Pólya, G., Szegő, G., \textit{Isoperimetric Inequalities in Mathematical Physics}, Annals of Mathematics Studies. Princeton, N.J. (1951) Princeton University Press.

\bibitem{Riesz} Riesz, F., \textit{Sur une inégalité intégrale}. Journal of the London Mathematical Society. 5 (3) (1930) pp. 162–168.

\bibitem{Struwe} Struwe, M., \textit{A global compactness result for elliptic boundary value problems involving limiting nonlinearities}, Math. Z., 187 (1984), pp. 511-517.

\bibitem{Yap} Yap., L. Y. H., \textit{Some remarks on convolution operators and L(p,q) spaces}. Duke Math. J., 36 (1969), pp. 647–658.

\end{thebibliography}
\end{document}